\documentclass[%
 reprint, 
 amsmath,amssymb,
 aps,
]{revtex4-1} 

\usepackage{dcolumn}
\usepackage{bm}

\usepackage[dvipdfmx]{graphicx}
\usepackage{amssymb,amsfonts,amsmath,amsthm}
\usepackage{mathrsfs}
\usepackage{array}

\usepackage{algorithmic}
\usepackage[boxed,commentsnumbered]{algorithm2e}

\newtheorem{thm}{Theorem}[]
\newtheorem{prop}{Proposition}
\newtheorem{definition}[]{Definition}

\newtheorem{question}[]{Question}

\def\diam{\mathrm{diam}}
\newcommand{\floor}[1]{\lfloor #1 \rfloor}
\newcommand{\ceil}[1]{\lceil #1 \rceil}
\newcommand{\D}[0]{\ \mathrm{d}}
\newcommand{\expectation}[2]{ \mathbb{E}_{#1}\left[ #2 \right] }
\newcommand{\expf}[2]{ \overline{ #2 }^{( #1 )} }

\begin{document}

\preprint{APS/123-QED}

\title{On the Estimation of Pointwise Dimension} 
\thanks{This study was supported by the NeuroCreative Lab, Grant-in-Aid for Scientific Research B No. 23300099, and Grant-in-Aid for Exploratory Research No. 25560297.}%



\author{Shohei Hidaka}
 \homepage{http://www.jaist.ac.jp/~shhidaka}
\affiliation{
 School of Knowledge Science\\
Japan Advance Institute of Science and Technology
}%
\author{Neeraj Kashyap}
 \homepage{http://www.jaist.ac.jp/~shhidaka}
\affiliation{
 School of Knowledge Science\\
Japan Advance Institute of Science and Technology
}%


\date{\today}

\begin{abstract}
Our goal in this paper is to develop an effective estimator of fractal dimension.
We survey existing ideas in dimension estimation, with a focus on the currently popular method of Grassberger and Procaccia for the estimation of correlation dimension.
There are two major difficulties in estimation based on this method.
The first is the insensitivity of correlation dimension itself to differences in dimensionality over data, which we term {\em dimension blindness}. 
The second comes from the reliance of the method on the inference of limiting behavior from finite data.

We propose pointwise dimension as an object for estimation in response to the dimension blindness of correlation dimension.
Pointwise dimension is a local quantity, and the distribution of pointwise dimensions over the data contains the information to which correlation dimension is blind.
We use a ``limit-free'' description of pointwise dimension to develop a new estimator.
We conclude by discussing potential applications of our estimator as well as some challenges it raises.

 \end{abstract}

\maketitle

\tableofcontents

\section{Introduction}
Dimension has become an important tool in the study of dynamical systems.
Its importance arises from the attempt to understand dynamical sytems by their attractors.
The dimensions of these attractors provide useful invariants of the systems in question.
There are many notions of dimension that one can consider in this context. 
The most well-known of these is Hausdorff dimension. Hausdorff dimension, however, is difficult to estimate numerically and this has given rise to many of the other conceptions.
\\

Often, one can associate with ergodic systems certain probability measures supported on attractors of interest.
There is a dimension theory for probability measures, roughly analogous to the classical dimension theory for sets, which one can bring to bear in the context of such systems.
This is the setting for this paper. In particular, we focus on the problem of estimating the {\it pointwise dimension distribution} of a  probability measure by sampling from it.
\\

In Section \ref{sec-MathematicalBackground}, we discuss various mathematical notions of dimension and the relationships between them.

In Section \ref{sec-Estimation}, we discuss the current most popular technique for estimating fractal dimension -- that of Grassberger and Procaccia \cite{Grassberger1983}.

This discussion motivates our derivation of our estimator of pointwise dimension, which we present in Section \ref{sec-Estimator}.

In Section \ref{sec-BasicAnalysis}, we establish a certain degree of reliability of the estimator proposed in Section \ref{sec-Estimator} by performing a dimensional analysis of some test data.

Finally, in Section \ref{sec-Future}, we discuss potential applications of our estimator as well as certain theoretical questions arising from our considerations.


\section{Mathematical Background \label{sec-MathematicalBackground}}

\subsection{Pointwise Dimension \label{sec-PointwiseDimension}}

Very generally, the purpose of any notion of dimension is to tell apart certain types of objects given a certain notion of equivalence of such objects.
Pointwise dimension tells apart Borel probability measure on metric spaces which are not equivalent by a locally bi-Lipschitz injection.
\\

To be precise let $\mu_{1}$ and $\mu_{2}$ be Borel probability measures on the metric spaces $( X_{1}, \rho_{1} )$ and $( X_{2}, \rho_{2} )$ respectively. For a point $x$ in a metric space $(X, \rho)$ and a number $\epsilon > 0$, let us write
\[
 B( x, \epsilon ) := \{ y \in X \mid \rho( y, x ) < \epsilon \}. 
\]
A map $f: X_{1} \rightarrow X_{2}$ is called {\it locally bi-Lipschitz} if for all $x \in X_{1}$ there exists $r > 0$ and a constant $K \ge 1$ such that for all $u, v \in B( x, r ) $ we have
\[
 \frac{1}{K} \rho_{1}( u, v ) \le \rho_{2}( f(u), f(v) ) \le K \rho_{1}( u, v ).
\]
Let us say that the measures $\mu_{1}$ and $\mu_{2}$ are {\it equivalent} if there exists a locally bi-Lipschitz injection $f: X_{1} \rightarrow X_{2}$ such that for all $\mu_{2}$-measurable $U \subset X_{2}$ we have 
\begin{equation}
 \mu_{1}( f^{-1}( U ) ) = \mu_{2}( U ).
  \label{eq-PointwiseEquivalence} 
\end{equation}
Pointwise dimension is a source of invariants of such probability measures under this notion of equivalence.
Rather than prove this general statement, we will illustrate this in a very simple situation which motivates the concept of pointwise dimension.
The illustration will contain all the ideas require for the general proof.
\\

Consider the unit interval $I$ along with its uniform measure $\nu_{1}$ and the unit square $I \times I$ along with its uniform measure $\nu_{2} = \nu_{1} \times \nu_{1}$, both under their natural Euclidean metrics.
The ``pointwise dimension'' method of telling these measures apart is to observe that, for $x \in I$, we have 
\[
 \nu_{1}( B( x, \epsilon ) ) \sim \epsilon, 
\]
whereas, for $y \in I \times I$, we have
\begin{equation}
 \nu_{2}( B( y, \epsilon ) ) \sim \epsilon^{2}.
\label{eq-SquareRate}
\end{equation} 
If there were an equivalence $f$ between $\nu_{1}$ and $\nu_{2}$, then we would have 
\[
 B\left( x, \frac{1}{K} \epsilon \right) \subseteq f^{-1}\left( B( f( x ), \epsilon ) \right) \subseteq B\left( x, K \epsilon \right).
\]
This would force
\[
 \nu_{1}( f^{-1}\left( B( f( x ), \epsilon ) \right) ) \sim \epsilon, 
\]
which would contradict (\ref{eq-SquareRate}).
\\

For a general Borel probability measure, growth rates such as the ones we computed above for $\nu_{1}$ and $\nu_{2}$ may not exist at all points or even at any point in the corresponding metric space.
This motivates the following definition: 
\begin{definition}
\label{def-PointwiseDimension}
Let $\mu$ be a Borel probability measure on a metric space $( X, \rho)$.
The lower pointwise dimension of $\mu$ at $x$ is defined to be 
\begin{equation}
 \underline{D}_{\mu}(x) := \liminf_{\epsilon \rightarrow 0} \frac{ \log \mu( B( x, \epsilon ) ) }{ \log \epsilon }.
\label{eq-LowerPointwiseDimension}
\end{equation}
The upper pointwise dimension of $\mu$ at $x$ is defined to be 
\begin{equation}
 \overline{D}_{\mu}(x) := \limsup_{\epsilon \rightarrow 0} \frac{ \log \mu( B( x, \epsilon ) ) }{ \log \epsilon }.
\label{eq-UpperPointwiseDimension}
\end{equation}
If these two quantities agree, we call the common value the pointwise dimension of $\mu$ at $x$ and denote it $D_{\mu}( x )$.
\end{definition}

The calculations we performed highlight two properties of pointwise dimension: 
\begin{enumerate}
 \item Pointwise dimension is invariant under bi-Lipschitz injections.
 \item If $\mu$ is absolutely continuous with respect to Lebesgue measure on $\mathbb{R}^{N}$, then the pointwise dimension of $\mu$ at any point in its support is $N$.
\end{enumerate}

Let us now see how pointwise dimension behaves for other classes of measures.
We begin by considering the Cantor measure -- this is the measure resulting from the Cantor middle-thirds construction if at each scale we choose one of the remaining intervals uniformly.
To be precise, for an interval $I$, let $u_{I}$ denote the uniform measure on $I$. At the $n^{\text{th}}$ stage of the middle third construction, there are $2^{n}$ admissible interval $I_{1}^{(n)}, I_{2}^{(n)}, \hdots, I_{2^{n}}^{(n)}$. Let us write 
\[
 \xi_{n} := \frac{1}{2^{n}} \sum_{j = 1}^{2^{n}} u_{ I_{j}^{(n)} }. 
\]
There is a limiting measure $\xi$ for the family $\{ \xi_{n} \}$ which we call the Cantor measure.
$\xi$ is supported on the Cantor set $\mathcal{C}$.
\\

For $x \in \mathcal{C}$ and $\epsilon > 0$, let us first calculate $\xi( B( x, \epsilon ) )$.
Note that the scale of this ball in terms of the middle-thirds construction is given by $\log_{1/3} ( \epsilon ) $. 
This means that 
\[
 2^{ - \ceil{ \log_{1/3} ( \epsilon ) } } \le \xi( B( x, \epsilon ) ) \le 2^{ - \floor{ \log_{1/3} ( \epsilon ) }}.
\]
This shows that
\[
  \frac{ \log( \xi( B( x, \epsilon ) ) ) }{ \log( \epsilon ) } \sim \frac{ \log_{3}( \epsilon ) \cdot \log( 2 ) }{ \log( \epsilon ) },
\]
which in turn proves that $D_{\xi}( x )$ exists and is equal to $\frac{ \log( 2 ) }{ \log( 3 ) }$.
\\

Of course, the Cantor measure is a special case of a much more general class of measures which we call {\it Cantor-like measures}. This class of measures was first studied by Billingsley in the context of dimension theory \cite{Billingsley1960,Billingsley1961}. 
We restrict our attention here to the case of Borel probability measure on $\mathbb{R}$ and essentially reproduce Billingsley's definition here. It is possible to give a more modern definition in terms of iterated function systems.\\ 

Let $k > 0$ be an integer and let $\pi = (p_{1}, p_{2}, \hdots, p_{k})$ be a probability vector.
Divide the unit interval $I$ into $k$ subintervals of equal length $I_{1}, I_{2}, \hdots, I_{k}$.
Using our notation for the uniform measure from before, we write 
\[
 \xi_{\pi, 1} := \sum_{j = 1}^{k} p_{j} u_{ I_{j} }.
\]
We can further subdivide each of the intervals $I_{j}$ into $k$ subintervals of equal length $I_{j,1}, I_{j, 2}, \hdots, I_{j, k}$.
This allows us define 
\[
 \xi_{ \pi, 2 } := \sum_{ j_{1} = 1 }^{ k } p_{ j_{1} } \sum_{ j_{2} = 1 }^{ k } p_{ j_{2} } u_{ I_{ j_{1}, j_{2} } }.
\]
Iterating this procedure $n$ times yields a probability measure $\xi_{\pi, n }$.
The family $\{ \xi_{ \pi, n } \}$ has a limiting measure which we denote $ \xi_{\pi} $ and which generalizes our earlier construction of the Cantor measure. We will now calculate the pointwise dimension of such a measure $\xi_{ \pi }$ at a point $x \in I$.
\\

A point $x \in I$ is uniquely specified by a sequence of integers $\{ j_{i} \}$ as 
\[
 \{ x \} = \bigcap_{ i = 1 }^{ \infty } I_{ j_{1}, j_{2}, \hdots, j_{i} }.
\]
Let $\epsilon > 0$. The scale of the ball $B( x, \epsilon )$ in terms of the construction of $\xi_{\pi}$ is given by $\log_{1/k}( \epsilon )$. In fact, if we write $s( \epsilon) := \floor{ \log_{1/k}( \epsilon ) }$ and $S( \epsilon ) := \ceil{ \log_{1/k}( \epsilon ) }$
\begin{equation}
  \prod_{ i = 1 }^{ S( \epsilon ) } p_{ j_{i} } \le \xi_{\pi}( B( x, \epsilon ) ) \le 
\prod_{ i = 1 }^{ s( \epsilon ) } p_{ j_{i} }.
\label{eq-XiPiSqueeze}
\end{equation}
Let us count, for $1 \le j \le k$, 
\[
 f_{n}( j ) := | \{ i \mid 1 \le i \le n, j_{i} = j \} |.
\]
From (\ref{eq-XiPiSqueeze}), we see that 
\begin{equation}
\frac{ \log( \xi_{\pi}( B( x, \epsilon ) )  ) }{ \log( \epsilon ) }
\sim
 \sum_{ j = 1 }^{ k } f_{ S( \epsilon ) }( j ) \frac{ \log( p_{j} ) }{ \log( \epsilon ) }.
\label{eq-CLSFrequency} 
\end{equation}
Let us denote by $A_{\pi}$ the set of $x \in I$ for which, as $\epsilon \rightarrow 0$,
\[
  \frac{ f_{ S( \epsilon ) }( j ) }{ \log( \epsilon ) } \rightarrow  \frac{p_{j}}{ \log( 1/ k ) }.
\]
Then $\xi_{\pi}( A_{\pi} ) = 1$ by the strong law of large numbers.
Therefore, almost everywhere with respect to the measure $\xi_{\pi}$, we have
\begin{equation}
 D_{\xi_{\pi}}( x ) = - \frac{1}{ \log( k ) }\sum_{j = 1}^{k} p_{j} \log\left( p_j \right). 
\end{equation}



Furthermore, from (\ref{eq-CLSFrequency}), we can see that there are infinitely many points in the support of $\xi_{\pi}$ at which its pointwise dimension is not defined.
We illustrate this idea with the following example: 

Put $n_{1}:=1$. For each integer $m > 1$, define 
\[
 n_{m} := \min\left\{ n \mid \frac{n-n_{m-1}}{n} > 1 - \frac{1}{2^{m}} \right\}.
\]
Let $\mathcal{I}$ denote the vector space of measures generated by $u_{ [a,b] }$ for $ a < b$ in $\mathbb{R}$.
Let us define $S: \mathcal{I} \rightarrow \mathcal{I}$ by 
\[
 S u_{[a,b]} := \frac{1}{2}\left( u_{ [a, \frac{b-a}{5}] } + u_{ [b - \frac{b-a}{5}, b] } \right).
\]
Let us define $T : \mathcal{I} \rightarrow \mathcal{I}$ by 
\[
 T u_{[a,b]} := \frac{1}{3}\left( u_{ [a, \frac{b-a}{5}] } + u_{ [a + 2 \frac{b-a}{5}, a + 3 \frac{b-a}{5} ] } + u_{ [b - \frac{b-a}{5}, b] } \right).
\]
Finally, put
\[
 \mu_{2m+1} := S^{n_{2m+1}} T^{n_{2m}} \cdots T^{n_{2}}S^{n_{1}} \cdot u_{ [0,1] },
\]
and
\[
 \mu_{2m} := T^{n_{2m}} S^{n_{2m-1}} \cdots T^{n_{2}}S^{n_{1}} \cdot u_{ [0,1] }.
\]
Then each measure $\mu_{k}$ is a Borel probability measure and the family $\{ \mu_{k} \}$ converges to a Borel probability measure $\mu$.
It is easy to see that, for every point $x$ in the support of $\mu$, we have 
\[
 \underline{D}_{\mu}( x ) = \frac{1}{2}\frac{ \log( 2 ) }{ \log( 5 ) }, 
\]
whereas
\[
 \overline{D}_{\mu}( x ) = \frac{1}{3}\frac{ \log( 3 ) }{ \log( 5 ) }, 
\]
Therefore, at no point in the support of $\mu$ does it have a well-defined pointwise dimension.
More serious examples of such pathological behavior have been given by Ledrappier and Misiurewicz \cite{Ledrappier1985}, and Cutler \cite{Cutler1990}.  
\\

Although there is much left to be said on the matter of pointwise dimension, we end this section here hoping that we have given the readers enough of sense of this quantity that he is comfortable working with it.

%


\subsection{Other notions of dimension \label{sec-OtherDimensions}}
\subsubsection{Hausdorff dimension \label{sec-HausdorffDimension}}
One of the most well-known notions of dimension is that of Hausdorff \cite{Hausdorff1918}.
The notion of Hausdorff dimension is built upon the notion of a Hausdorff measure.

Let $E \subseteq \mathbb{R}^{N}$, and let $\epsilon > 0$.
An $\epsilon$-{\it covering} of $E$ is a countable collection $\{E_{k}\}_{k}$ of sets such that $E \subseteq \cup_{k}E_{k}$ and $\sup_{k} \diam (E_{k}) \le \epsilon$. 
For $\epsilon > 0$, put 
\begin{equation}
 \label{eq-HausdorffMeasureFixedScale}
 H^{ \alpha }( E, \epsilon ) := \inf_{ \{ E_{k} \} }  \left\{
 \sum_{k} \diam ( E_{k} )^{ \alpha } \Big| \left\{ E_{k} \right\}_{k} 
\right\},
\end{equation}
where the infimum is over $\epsilon$-coverings $\{ E_{k} \}$ of $E$.
Then Hausdorff $\alpha$-measure ($\alpha > 0$) is defined by 
\begin{equation}
 H^{ \alpha }( E ) := \lim_{\epsilon \rightarrow 0} H^{\alpha}( E, \epsilon ).
\label{eq-HausdorffMeasure}
\end{equation}
Note that this limit either exists or is $\infty$ because $H^{\alpha}( E, \epsilon )$ is a non-increasing function of $\epsilon$.
\\

Observe that, for $\alpha < \beta$, we have the inequality 
\[
 H^{\beta}( E, \epsilon ) \le \epsilon^{ \beta - \alpha } H^{\alpha}( E, \epsilon ).
\]
Suppose that $ H^{\alpha}( E ) \in ( 0, \infty )$. Then allowing $\epsilon \rightarrow 0$ in the above inequality forces $H^{\beta}( E ) = 0$.
Therefore there exists a unique critical exponent $\alpha$ 
such that
\begin{equation}
H^{\beta}( E )= \left \{
\begin{array}{ll}
0, & \beta > \alpha \\
\infty, & \beta < \alpha
\end{array}
\right..
\label{eq-HausdorffExponent}
\end{equation}
\begin{definition}
 For a set $E \subset \mathbb{R}^{N}$, we call the unique exponent $\alpha$ satisfying (\ref{eq-HausdorffExponent}) its {\em Hausdorff dimension}, and denote it by $D_{H}( E )$.
\end{definition}
This exponent $\alpha$ is defined to be the Hausdorff dimension $D_{H}( E )$ of the set $E$. 
Hausdorff dimension has the following properties:
\begin{enumerate}
 \item Hausdorff dimension is invariant under locally bi-Lipschitz injections,
 \item for a decomposition of $E$ into countably many sets $E = \cup_{k}E_{k}$, we have $D_{H}( E ) = \sup\{ D_{H}( E_{k} )\}$,
 \item the Hausdorff dimension of a ball in $\mathbb{R}^{N}$ is $N$. 
\end{enumerate}
Broadly speaking, Hausdorff dimension provides a measure of complexity for sets which is more finely stratified than, say, Lebesgue measure. For example, the Cantor set has Lebesgue meaure zero but Hausdorff dimension $\log(2)/\log(3)$. This makes Hausdorff dimension a very useful tool in the sciences.
We now discuss certain developments around Hausdorff dimension in the numerical study of dynamical systems.
\\

Although there is no rigorous criterion by which one may classify a set as fractal, a fractional Hausdorff dimension is one important indicator of such structure. 
As a consequence, Hausdorff dimension has long been used to detect chaotic behavior in dyanamical systems (see, for example, Ott \cite{Ott2002}).
For example, the Hausdorff dimension of a smooth attractor is a geometric invariant and therefore one may use it to classify such attractors.
\\

In practice, Hausdorff dimension is difficult to estimate from data. One of the reasons for this is the definition of the Hausdorff $\alpha$-measures themselves (Eq. \ref{eq-HausdorffMeasure}). Given a set $E$, effective estimation of $H^{\alpha}(E)$ requires one to produce an $\epsilon$-covering of $E$ which is close to optimal. The problem with this is that at different places in $E$ such an $\epsilon$-covering would employ sets of different diameters.
Furthermore, it is difficult to identify the critical exponent $\alpha$ with only estimates of the $\alpha$-measures.
This is one reason that Hausdorff dimension is difficult to estimate.
Another reason comes from property 2 of Hausdorff dimension listed above. The set of points in $E$ which dictate its Hausdorff dimension may be relatively ``small'' which presents difficulty in the estimation.
\\

Most numerical estimators of dimension have sought to evade these difficulties by estimating instead quantities which merely provides approximations to Hausdorff dimension.
The most well-known of these approximations are capacity or box-counting dimension, correlation dimension, and information dimension. We now present each of these in turn along with some estimators. 

\subsubsection{Box-counting dimension}
The idea behind box-counting dimension is that the dimension of a set is related to asymptotic behavior of the number of cubes of side length $\epsilon$ required to cover it as $\epsilon \rightarrow 0$. 
For example, for $\epsilon > 0$, one require $\epsilon^{-N}$ cubes of side length $\epsilon$ to cover the unit cube in $\mathbb{R}^{N}$. This allows us to recover the dimension $N$ as 
\[
 N = \lim_{\epsilon \rightarrow 0} \frac{ \log(\epsilon^{-N}) }{ \log(\epsilon^{-1})}.
\]
This motivates the following general definition:
\begin{definition}
\label{def-Box-counting}
Let $E \subset \mathbb{R}^{N}$. Let us denote by $B( E, \epsilon )$ the minimal number of cubes (or `boxes') of side length $\epsilon$ require to cover $E$.

We define the lower box-counting dimension of $E$ by
\[
 \underline{D}_{B}( E ) := \liminf_{ \epsilon \rightarrow 0 } \frac{ \log( B( E, \epsilon ) ) }{ \log( \epsilon^{-1} ) }.
\]
Similarly, we define the upper box-counting dimension of $E$ by
\[
 \overline{D}_{B}( E ) := \limsup_{ \epsilon \rightarrow 0 } \frac{ \log( B( E, \epsilon ) ) }{ \log( \epsilon^{-1} ) }.
\]
If these quantities agree, we call the common value the {\em box-counting dimension} of $E$, and denote it by $D_{B}( E )$. 
\end{definition}

If, in the definition of Hausdorff dimension, we take the infimum in (\ref{eq-HausdorffMeasureFixedScale}) over covers $\{ E_{k} \}$ of $E$ where each $E_{k}$ has a diameter exactly $\epsilon$, and if the corresponding limit (\ref{eq-HausdorffMeasure}) exists, we recover the box-counting dimension of $E$.
\\

Incidentally, the above relationship between the definitions of Hausdorff dimension and box-counting dimension shows that, for any $E \subset \mathbb{R}^{N}$,
\begin{equation}
 \label{eq-HDBD}
 D_{H}( E ) \le \underline{D}_{B}( E ).
\end{equation}

Frequently, it is the case that the inequality in (\ref{eq-HDBD}) is an equality.
In a sense, it is the situations in which the inequality is strict which make box-counting dimension a poor notion of dimension. 
For example, box-counting dimension does not satisfy anything like Property (2) of Hausdorff dimension -- 
one cannot calculate the box-counting dimension of a set using a countable decomposition of that set.
As an illustration, let $Q := \mathbb{Q} \cap I$. As $Q$ is dense in $I$, any covering of $Q$ by $\epsilon$-balls must also cover $I$. Therefore, $D_{B}( Q ) = D_{B}( I ) = 1$. 
On the other hand, if we enumerate $Q$ and write 
\[
 Q = \bigcup_{ k } \{ q_{k} \}, 
\]
we have $D_{B}( \{ q_{k} \}) = 0$ for all $k$.
\\

Despite its flaws, box-counting dimension remains popular because it is so easy to estimate. 
We will discuss the estimation of box-counting dimension in Section \ref{sec-Estimation}.

\subsubsection{Correlation Dimension \label{sec-CorrelationDimension}}
Grassberger and Procaccia, in \cite{Grassberger1983}, proposed a method of estimation of fractal dimension. 
However, it was not originally clear what kind of dimension they were estimating. Since then this notion has been formalized to produce what we know today as {\em correlation dimension} -- see, for example, Pesin \cite{Pesin1993}.
In this section, we give a brief account of this historical development.
\\

The basic idea of Grassberger and Procaccia was very similar to the idea which motivated our description of pointwise dimension in Section \ref{sec-PointwiseDimension}. 
The problem that they considered was that of estimating some kind of fractal dimension for an attractor in a dynamical system given a finite number of iterates of some initial values.
They reasoned that the rate at which the cardinality of intersections of $\epsilon$-balls around the data points with the data decayed would provide a reasonable notion of dimension provided that such a rate was well-defined.
\\

To be precise, given the data
\[
 E = \{ x_{1}, x_{2}, \hdots, x_{n} \} \subset \mathbb{R}^{N},
\]
let us fix a norm $\| \cdot \|$ on $\mathbb{R}^{N}$, and define 
the {\it correlation sum}
\begin{equation}
 C( E, \epsilon ) := \frac{1}{n}\sum_{i=1}^{n}
\left[ \frac{1}{n-1}\sum_{ j \neq i } 
\chi_{ [0, \epsilon) }
\left( 
\| x_{i} - x_{j} \|
\right)
\right],
\label{eq-CorrelationSum}
\end{equation}
where $\chi_{ [0, \epsilon) }$ denotes the characteristic function of the interval $[0, \epsilon )$.
Grassberger and Procaccia \cite{Grassberger1983} suggested that estimating the rate of decay of $C( E, \epsilon )$ with $\epsilon$ would provide a good notion of dimension for the set $E$.
While this technique was in some sense very easy to apply numerically, it lacked a rigorous mathematical foundation -- there was no known notion of dimension that the Grassberger-Procaccia method (GP) could be proved to provide an estimate of in general. Implicitly, Grassberger and Procaccia had introduced a completely new notion of dimension -- correlation dimension. The previously cited paper of Pesin \cite{Pesin1993} develops correlation dimension from a vague concept to a rigorous quantity.
\\

One difficulty in formalizing the concept of correlation dimension is that correlation sums (\ref{eq-CorrelationSum}) 
are defined for finite sets of data.
For example, this prevents us from defining the correlation dimension as simply a limit of the correlation sums as $\epsilon \rightarrow 0$. In fact, this particular point introduces much difficulty even in numerical estimation. We will discuss this in detail in Secton \ref{sec-Estimation}.
For the purposes of rigorous formulation, another problem with the finiteness of data is that it is not clear what kind of underlying process of data generation one should consider when trying to define correlation dimension.
In this paper, we will assume that the data is sampled according to a probability measure $\mu$, and therefore we define the correlation dimension {\em of} $\mu$.
\\

Let $\mu$ be a probability measure, and let $\left\{ X_{j} \right\}$ be a sequence of independent random variables with distribution $\mu$. Let us define 
\[
 E_{n} := \{ X_{1}, X_{2}, \hdots, X_{n} \}.
\]
Then, for a fixed $\epsilon > 0$, the limit 
\begin{equation}
 \label{eq-CorrelationLimit}
 C( \mu, \epsilon ) :=  \lim_{ n \rightarrow \infty } C( E_{n}, \epsilon )
\end{equation}
exists almost surely by the strong law of large numbers, and is equal to the probability that $\| X_{1} - X_{2} \| < \epsilon$.
\begin{definition}
 \label{def-CorrelationDimension}
The limit
\[
 D_{C}( \mu ) := \lim_{\epsilon \rightarrow 0} \frac{ \log( C( \mu, \epsilon ) ) }{ \log( \epsilon ) }, 
\]
if it exists, is called the {\em correlation dimension} of $\mu$.
\end{definition}
If $\mu$ is a Borel probability measure on $\mathbb{R}^{N}$, then let us write for $x \in \mathbb{R}^{N}$ 
\[
 F_{\mu, \epsilon}( x ) := \mu( B( x, \epsilon ) ).
\]
In this case, we have 
\[
 C( \mu, \epsilon ) = \| F_{\mu, \epsilon} \|_{1},
\]
and so 
\[
 D_{C}( \mu ) = \lim_{\epsilon \rightarrow 0} \frac{ \log( \| F_{\mu, \epsilon} \|_{1} ) }{ \log( \epsilon ) }.
\]
This suggests that there is an entire family of dimension related to the correlation dimension.
These dimensions are called the $q$-generalized dimensions, were first introduced by Hentschel and Procaccia \cite{Hentschel1983} in a manner similar to that in which correlation dimension was introduced by Grassberger and Procaccia \cite{Grassberger1983}.
Formally, we define the $q$-generalized dimension of $\mu$, if it exists, to be 
\begin{equation}
 D_{C,q}( \mu ) := \lim_{\epsilon \rightarrow 0} \frac{ \log( \| F_{\mu, \epsilon} \|_{q-1} ) }{ \log( \epsilon ) }.
 \label{eq-GeneralizedDimension} 
\end{equation}
We take the norm in $L^{q-1}( \mu )$ above in order to be consistent with the notation used in the numerical literature.
The reason for this convention is apparent from the derivation of Hentschel and Procaccia \cite{Hentschel1983}.
In this paper, we will be concerned mainly with the correlation dimension although much of our discussion of correlation dimension carries over directly to the case of the $q$-generalized dimension.
\\

Note that, in definition, there is quite a bit of similarity between the correlation dimension of $\mu$ (Definition \ref{def-CorrelationDimension}) and its pointwise dimension at some generic point $x$ in its support (Definition \ref{def-PointwiseDimension}).
For example, Cutler \cite{Cutler1993} defines the {\em average pointwise dimension} of a Borel probability measure $\mu$ on $\mathbb{R}^{N}$ by 
\[
 D_{P}( \mu ) := \int_{\mathbb{R}^{N}} D_{\mu}( x ) \D\mu(x).
\]
She points out that, while $D_{P}( \mu )$ reflects the average of the local rates of decay of the measure of $\epsilon$-balls around points in the support of $\mu$, $D_{C}( \mu )$ reflects the rate of decay of the average measure of $\epsilon$-balls around points in the support of $\mu$.
This seems to have been a source of confusion in the past.
We urge the reader to be wary both in reading about and applying these concepts.
\\

To expand upon this, suppose that for each $i = 1, 2$ we have a Borel probability measure $\mu_{i}$ on $\mathbb{R}^N$ with associated  $\epsilon_{i} > 0$ such that for each $0 < \epsilon < \epsilon_{i}$
we have 
\[
 \mu_{i}( B( x, \epsilon ) ) \sim \epsilon^{ \alpha_{i} }.
\]
Let us further assume that the supports of $\mu_{1}$ and $\mu_{2}$ are contained in disjoint open sets in $\mathbb{R}^{N}$ and that $\alpha_{1} < \alpha_{2}$.
Finally, let us define 
\[
 \mu = \frac{1}{2}\mu_{1} + \frac{1}{2}\mu_{2}.
\]
As we will see in (\ref{eq-CorrelationPointwise}), for each $i = 1, 2$, we have 
\[
 D_{C}( \mu_{i} ) = \alpha_{i}.
\]
Now, we also have, for $\epsilon > 0$,
\[
 C( \mu, \epsilon ) = \frac{1}{2}C( \mu_{1}, \epsilon ) + \frac{1}{2}C( \mu_{2}, \epsilon ).
\]
Thus, for $0 < \epsilon < \min( \epsilon_{1}, \epsilon_{2} ) $, we have 
\[
 C( \mu, \epsilon ) \sim \frac{1}{2}\epsilon^{ \alpha_{1} } \left( 1 + \epsilon^{ \alpha_{2} - \alpha_{1} } \right).
\]
Therefore, 
\begin{equation}
 \label{eq-DimensionBlindness}
  D_{C}( \mu ) = \lim_{\epsilon \rightarrow 0} \frac{ \log( \epsilon^{ \alpha_{1} } ) + \log \left( 1 + \epsilon^{ \alpha_{2} - \alpha_{1} } \right) }{\log( \epsilon )} = \alpha_{1}.
\end{equation}
This already highlights one of the major problems with correlation dimension --
it cannot even be used to tell apart a simple mixture of measures such as $\mu$ from its lowest dimensional component!
We call this the {\em dimension blindness} of correlation dimension.
This poses a serious problem with relying upon estimators of correlation dimension in the analysis of data.
We seek to address this problem with our estimator.


\subsection{Dimensions of Measures \label{sec-MeasureDimensions}}

Although Hausdorff dimension provides a measure of complexity of {\em sets}, as we have defined it, it is easy to derive from it a notion of dimension for measures. We follow here the development of Cutler \cite{Cutler1993}.
\begin{definition}[]
\label{def-HausdorffDimensionOfMeasure}
 Let $\mu$ be a Borel probability measure on $\mathbb{R}^{N}$. We define the Hausdorff dimension distribution $\mu_{H}$ of $\mu$, which is a Borel probability measure on $[0, N]$, by 
\[
 \mu_{H}( [0, \alpha] ) := \sup\{ \mu( E ) \mid  D_{H}( E ) \le \alpha \}.
\]
We define the Hausdorff dimension of $\mu$ to be 
\[
 D_{H}( \mu ) := \inf\{ \alpha \in [ 0, N ] \mid  \mu_{H}( [0, \alpha] ) = 1 \}.
\]
\end{definition}

At this point, we have described in some detail three notions of dimension associated with a Borel probability measure $\mu$ on $\mathbb{R}^{N}$ -- pointwise dimension, correlation dimension, and Hausdorff dimension.
A lot is known about the relationships between these three dimensions, particularly when $\mu$ has constant pointwise dimension 
\[
 D_{\mu}( x ) = d, 
\]
for almost every $x$ in its support.
Young \cite{Young1982} proved that, when this is the case, we have 
\[
 D_{H}( \mu ) = d.
\]
Pesin \cite{Pesin1993} proved that, in the same case, we have
\begin{equation}
 D_{C}( \mu ) = d. 
 \label{eq-CorrelationPointwise}
\end{equation}
In fact, the results of Young and Pesin are a bit more general. 
Cutler has explored in greater detail the relationship between the distribution of pointwise dimensions of $\mu$ and its corresponding Hausdorff dimension distribution $\mu_{H}$: 
\begin{thm}[Cutler \cite{Cutler1993}, Theorem 3.2.2]
Let $\mu$ be a Borel probability measure on $\mathbb{R}^{N}$.
Let us consider the following subset  of the support of $\mu$:
\[
 D_{\alpha} := \{ x \mid \underline{D}_{\mu}( x ) \le \alpha \}.
\]
Then $D_{H}( D_{\alpha} ) \le \alpha$, and $\mu_{H}( [0, \alpha] ) = \mu( D_{\alpha} )$.
\end{thm}

To give a rough idea of the kind of reasoning required to prove these results, let us consider (\ref{eq-CorrelationPointwise}). 
Suppose that $\mu$ satisfies the following condition:
there exist constants $c, C  >0$  and $\epsilon_{0} > 0$ such that,
for all points $x$ in the support of $\mu$ and all $0 < \epsilon < \epsilon_{0}$, 
\begin{equation}
 \label{eq-UniformScaling}
 c \epsilon^{d} < \mu( B( x, \epsilon ) ) < C \epsilon^{d}. 
\end{equation}
Then it is easy to see that $D_{C}( \mu ) = d$ as, for $0 < \epsilon < \epsilon_{0}$, we have
\[
 c \epsilon^{d} < C( \mu, \epsilon ) < C \epsilon^{d}.
\]
For a general Borel probability measure $\mu$ with almost everywhere constant pointwise dimension $d$, there is no such garantee of uniformity in local scaling.
The trick to proving the general result is to observe that any such measure can be written as a limit of measures satisfying the condition (\ref{eq-UniformScaling}).
To see this, we can restrict $\mu$ to those points in its support at which the condition is satisfied for various values of $c, C,$ and  $\epsilon_{0}$. $\mu$ is a limit of such restrictions.
\\



The condition that $\mu$ has constant pointwise dimension almost everywhere in its support is not an exceptional one as far as dynamical data is concerned -- for example, if $\mu$ is ergodic for a map $f$, then this condition is satisfied under quite general constraints on $f$. For more details, refer to Cutler (\cite{Cutler1990,Cutler1992}).
We state here a simplified version of Cutler's results:
\begin{prop} 
 \label{prop-ErgodicityExactDimension}
 Let $f : X \rightarrow X $  be a locally bi-Lipschitz, ergodic map with ergodic measure $\mu$ for which the pointwise dimension exists almost everywhere. Then $\mu$ has constant pointwise dimension almost everywhere in its support.
\end{prop}

\begin{proof}
 Since $\mu$ is invariant under $f$, $f$ defines an equivalence between $\mu$ and itself in the sense of 
(\ref{eq-PointwiseEquivalence}) in Section \ref{sec-PointwiseDimension}. This tells us that the function $D_{\mu}(x)$ is $f$-invariant. 
As $f$ is ergodic, this means that $D_{\mu}(x)$ is constant.
\end{proof}

This can form the basis for a test of ergodicity. We will discuss this in Section \ref{sec-FutureDynamicalEstimation}. 


\section{Numerical Estimation \label{sec-Estimation}}

In this section, we demonstrate how one may use the idea of Grassberger and Procaccia \cite{Grassberger1983} to estimate the correlation dimension of a measure. We call such methods {\em GP estimators}. 
Currently, GP estimators seem to be the most frequently used estimators of fractal dimension in the applied sciences. We begin by presenting the data to which we will apply the estimator, along with a theoretical analysis.
We then describe how one would design a GP estimator for this data. 
We present the results of an analysis of the data with such an estimator. 
Finally, we conclude by discussing some of the drawbacks of GP estimators.
This will motivate our estimator of pointwise dimension, which we present in the next section.
\\

\subsection{The data \label{sec-GPData}}
We will use three measures in our demonstration of the GP estimator.
The first probability measure that we use is the Cantor measure $\xi$, which we defined in Section \ref{sec-PointwiseDimension}.
The second measure we use is 
\[
  \zeta := \frac{1}{2} \left( \xi + \xi' \right),
\]
where $\xi'$ is the measure derived from $\xi$ by mapping the interval $[0,1]$ linearly onto $[2, 7]$. 
This represents the image of $\xi$ under a piecewise linear transformation.
The third measure that we use is a mixture of two measures with different dimensional behavior:
\[
 \eta := \frac{1}{3} \delta_{0} \times \xi + \frac{2}{3} \left( \xi * \delta_{1}  \right) \times \left( \xi * \delta_{1} \right),
\]
where $\delta_{\alpha}$ denotes the point-mass at $\alpha \in \mathbb{R}$.
\\

Before we begin our analysis, we must calculate the correlation dimensions of each of these measures.
The calculation for the first two measures follows from (\ref{eq-CorrelationPointwise}) upon observing that the pointwise dimension of those measures is equal to $\log( 2 )/ \log( 3 )$ almost everywhere in their supports. 
For the measure $\eta$, observe that 
\[
 D_{C}( \delta_{0} \times \xi ) = \frac{ \log( 2 )}{ \log( 3 )},
\]
and 
\[
 D_{C}( \left( \xi * \delta_{1}  \right) \times \left( \xi * \delta_{1} \right) ) = 2 \frac{ \log( 2 )}{ \log( 3 )}.
\]
The calculation of the correlation dimension of $\eta$ essentially follows the calculation (\ref{eq-DimensionBlindness}), and we have 
\[
 D_{C}( \eta ) = \frac{ \log( 2 ) }{ \log( 3 ) }.
\]
Once again, dimension blindness rears its ugly head. Still, it is interesting to see how it manifests itself in numerical estimation.

\subsection{Grassberger-Procaccia Estimators}


Suppose that we are given a dataset $\mathcal{D}$ of $n$ points sampled from a Borel probability measure $\mu$ on $\mathbb{R}^{N}$, and that we wish to use this data to estimate the correlation dimension of $\mu$.
The basic idea in any GP estimator of correlation dimension is to use the correlation sums 
defined in (\ref{eq-CorrelationSum}) at various scales to perform this estimation according to Definition \ref{def-CorrelationDimension}.
Since the data is only available at a coarse resolution, some delicacy is required in inferring the limiting behavior of the correlation sums from the few which are available from the data.
\\

What one often does is evaluate correlation sums at various scales, decides which range of scales contains the most information about the correlation dimension, and finally use a linear model to estimate the quantity
\[
 \frac{\log( C( \mu, \epsilon ) )}{ \log( \epsilon ) }
\]
over this range.
This burden on the analyst to choose an informative range of scales is problematic as it can potentially be a large source of bias in the analysis.
We will experience the difficulty of making this choice in our analysis of the test data.



\subsection{Analysis \label{sec-GPAnalysis}}

 \begin{figure}[htb, clip]
   \begin{center}
    \includegraphics[width=1\linewidth, clip]{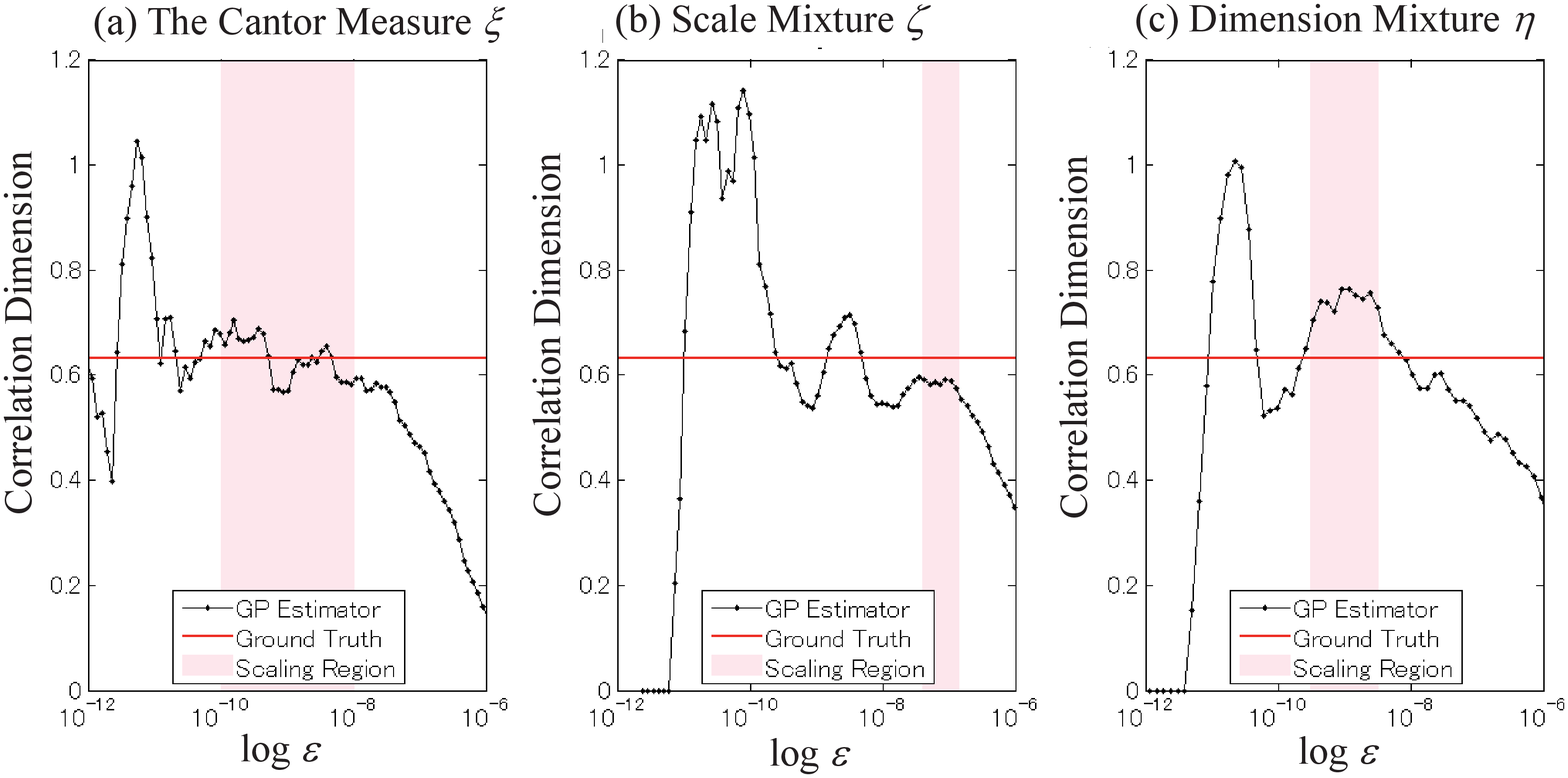}
    \caption{ The GP estimator (a) on a dataset sampled from the 1-dimensional Cantor set and (c) on a dataset sampled from an object which is mixture of 1- and 2-dimenisonal Cantor set. (b) The GP estimator on the same dataset as (a) transformed with a piecewise linear function.
    \label{fig-GPAnalysis} }      
    \end{center}
  \end{figure}



In Figure \ref{fig-GPAnalysis}, we present graphs of $\log( C( \mu, \epsilon ) )$ as functions of $\log( \epsilon )$ for each of the measure $\mu = \xi, \zeta, \eta$ of Section \ref{sec-GPData}.
We sampled 10,000 points from each measure up to a precision of 30 ternary digits for this analysis.
We have indicated our choices of informative regions on each graph by shading them. 
The horizontal line reflects the true correlation dimension of $\log( 2 )/ \log( 3 )$ in each graph.
\\

Observe that, after taking logs of the correlation sums and the parameters $\epsilon$, the informative range of scales is the one over which the regression line for the data is closest to horizontal and the total least square error is relatively small. It is difficult to state more precise criterion for choosing this region.
For example, in Figure \ref{fig-GPAnalysis}(b), there seems to be no clear choice of informative scaling region.
\\

From Figure \ref{fig-GPAnalysis}, we obtain the following estimates (the ground truth is $\log( 2 )/ \log( 3 ) = 0.63093$):
\begin{enumerate}
 \item $D_{C}( \xi ) \approx 0.6305 \pm 0.0411$, 
 \item $D_{C}( \zeta ) \approx 0.5859 \pm 0.0066$,  
 \item $D_{C}( \eta ) \approx 0.7411 \pm 0.0194$. 
\end{enumerate}

From Figure \ref{fig-GPAnalysis}(c), one may be tempted to claim that the effects of the mixture are observable -- there is a distinct bump in the graph over the range $10^{-11} < \log( \epsilon ) < 10^{-10}$.
Such a claim is not easily justifiable, however, for the same behavior is present in Figure \ref{fig-GPAnalysis}(a).
\\

Figure \ref{fig-GPAnalysis}(b) is of particular interest, for it shows that the GP estimates of correlation dimension are not sensitive even to piecewise linear transformations of the data.

\section{The Estimator \label{sec-Estimator}}

\subsection{ Guiding Principles \label{sec-Principle} }
We observed, in Section \ref{sec-Estimation}, some of difficulties involved in estimating dimension using a GP estimator. Broadly, these difficulties arise from two sources: 
\begin{enumerate}
 \item The dimension blindness of correlation dimension.
 \item The onus on the user to choose correct scale in analysis.
\end{enumerate}
These are important issues to address when designing a new estimator.
\\

Bearing in mind the dimension blindness of correlation dimension, it seems desirable to instead estimate a dimension which is more sensitive to the local structure of a given probability measure.
Therefore, we choose to estimate pointwise dimension.
This enables us to treat the dimension as a distribution, which has certain statistical advantages beyond the elimination of blindness of varying dimensions. This will become clear in our numerical analysis of the estimator.
\\

Overcoming the problem of the sensitivity of data to scale is much more difficult.
The reason that this sensitivity is such a problem in GP estimators is that Grassberger and Procaccia \cite{Grassberger1983} explicitly depend upon {\em limiting behavior} to produce good estimates.
Addressing this problem in the case of an estimator of pointwise dimension is not easy because pointwise dimension is fundamentally a local quantity.
We utilize the idea that distances of data points to their nearest neighbors contain information about local dimensions of the measure that the data is sampled from. 
We expand upon this idea in more detail in the next section.
\\

Finally, attempting to estimate pointwise dimension carries with it certain difficulties which are not present when estimating correlation dimension.
The principal problem is that to get a true sense of distribution of pointwise dimensions over the data can be computationally very expensive. 
As a result of this computational difficulty, although schemes to estimate pointwise dimension have previously been suggested (for example, Cutler \cite{Cutler1993}), there has been no large scale implementation that we are aware of. In our estimator, we utilize clustering techniques to mitigate this cost.
The idea is to identify data points at which the underlying measure seems to have similar scaling properties.
We discuss this in greater detail in Section \ref{sec-AlgorithmOverview} and \ref{sec-VBApproximation}. 
\\

Finally, one situation in which it was difficult to choose scales for analysis was when the data had been transformed by a locally bi-Lipschitz injection.
Although we have not emphasized this point very much in our discussion so far, it was a primary motivation in developing our estimator.
As pointwise dimension is purely local quantity, an estimator of pointwise dimension should naturally be less sensitive to the effects of such transformations.
In fact, we incorporate such insensitivity to local scaling in our estimator in the form of local density parameters,
this is certainly something that the reader should be aware of as the proceed through this paper.
\\

To summarize, there are four points that the reader should keep in mind throughout the rest of this paper:
\begin{enumerate}
 \item {\bf Pointwise dimension}: We estimate pointwise dimension, which is sensitive to differences in scaling behavior over the data.
 \item {\bf Limit-free description}: We utilize a limit-free description of pointwise dimension, which makes estimation more effective.
 \item {\bf Transformation invariance}: We introduce a local density parameter which gives our estimator the flexibility to cope with data which has been transformed by a locally bi-Lipschitz injection.
 \item {\bf Clustering}: We use clustering techniques to identify data points with similar pointwise dimension characteristics, which saves on computational cost.
\end{enumerate}

\subsection{A limit-free description of pointwise dimension \label{sec-LimitFree}}




In this section, we develop an estimator of pointwise dimension which utilizes the distances from data points to their nearest neighbors. 
The distributions of these distances are approximated by waiting time distributions for Poisson processes.
Such approximations provide an effective means of estimation for a large class of generating measures $\mu$.
Cutler and Dawson \cite{Cutler1989,Cutler1990NN} have similar results for different classes of measures using very different techniques.
\\

Let us begin by considering a Borel probability measure $\mu$ which satisfies the following {\em local uniformity condition}:

Suppose that the pointwise dimension $D_{\mu}( x )$ is defined at every point $x$ in the support of $\mu$.
Suppose further that, at every such point $x$, there is a constant $\delta_{x} > 0$ and $\epsilon' > 0$ such that, for all $0 < \epsilon < \epsilon'$, 
\begin{equation}
  \mu( B(x, \epsilon ) ) =  \delta_{x} \epsilon^{D_{\mu}( x )}.
 \label{eq-LocalUniformity}
\end{equation}


The main difficulty in obtaining the limit-free description of $D_{\mu}( x )$ is 
that it is difficult to make general statements about the infinitesimal behavior of $\mu$ at $x$.
We will show that, 
as far as measures satisfying the local uniformity condition are concerned,
one may obtain effective estimates of pointwise dimension by assuming 
that the infinitesimal behavior is that of a Poisson process.
\\

Let $\mu$ be a Borel probability measure on $\mathbb{R}^{N}$ satisfying the local uniformity condition (\ref{eq-LocalUniformity}) above,
and suppose we would like to estimate $D_{\mu}(x)$ given some data sampled from $\mu$.
Let us say that the question of whether or not there is a data point at distance $r$ from $x$ is determined by a Poisson process $Z(r)$ with rate $\delta_{x} r^{D_{\mu}( x )}$.
Explicitly $Z(r)$ counts the number $n$ of data points at distance up to $r$ from $x$, and has density function.
\begin{equation}
 \zeta( n ) = \frac{ \delta_{x}^{n} r^{nD_{\mu}( x )} }{ n! } \exp\left( - \delta r^{D_{\mu}( x )} \right)
 \label{eq-PoissonDistribution}
\end{equation}
Let us define, for $n > 0$, the waiting times
\[
 R_{n}( x ) := \sup\{ r > 0 \mid Z(r) < n \}.
\]
These are the distances from $x$ to its $n^{\text{th}}$ nearest data points.
The probability densities of the random variables $R_{n}( x )$ can be easily calculated from that of $Z( r )$ as 
\begin{equation}
 \phi_{n}( r ) = \frac{ D_{\mu}(x) \delta_{x}^{n} r^{nD_{\mu}(x) - 1} }{ (n-1)! } \exp\left( - \delta_{x} r^{ D_{\mu}(x) } \right).
\label{eq-NearestNeighborDistribution}
\end{equation}
The special case $n = 1$ has been derived in Cutler as a log gamma distribution (\cite{Cutler1992NN}; See also \cite{Cutler1993} Definition 7.2.1) and \cite{Badii1985}. 
The essentially same form as (\ref{eq-PoissonDistribution}) is also used by 
Nerenberg and Essex \cite{Nerenberg1990}
in the context of error analysis correlation dimension. 
\\

Since $\mu$ satisfies the local uniformity condition, it looks like a uniform measure in a small ball around $x$.
Specifically, it looks like a uniform measure which arises from the Poisson process described above by conditioning on the region where the data is sampled.
As a result, the distribution corresponding to (\ref{eq-NearestNeighborDistribution}) is a good approximation to the true nearest neighbor distribution at small scales.
This means that, if the data is fine enough, (\ref{eq-NearestNeighborDistribution}) provides an effective way of estimating the pointwise dimension.
As we will see in our analysis of our estimator, this sensitivity to scale of the approximation is not very marked.
\\

We end this section by showing how the limit-free description we have given above corresponds to a more conventional derivation of pointwise dimension.

Let us denote the distribution function of (\ref{eq-NearestNeighborDistribution}) by $F_{n}( r ) = \int_{0}^{r} \phi_{n}( r ) \D r$. Specifically, it is given by an incomplete gamma function:
\begin{equation}
 F_{n}( r ) = \int_{0}^{ \delta r^{ D_{\mu}(x) }} t^{ n - 1 } \exp( - t ) \D t.
\label{eq-NNCumulativeDistribution}
\end{equation}

In fact, $F_{n}( r )$ gives the same asymptotic scaling behavior as $\mu$.
Namely, using LHopital's rule, we get
\begin{equation}
\label{eq-PointwiseDimensionLimit}
 \lim_{r \rightarrow 0}\frac{1}{n}\frac{ \log F_{n}( r )}{ \log r } = D_{\mu}(x).
\end{equation}

As the densities $\delta_{x}$ are allowed to vary, and since we are trying to estimate the {\em distribution} of pointwise dimensions of $\mu$, it will be useful for us to distinguish between the distributions $F_{n}( r )$ corresponding to different choices of parameters in (\ref{eq-LocalUniformity}). Therefore, we consider the conditional distributions 
\[
 F_{n}( r \mid \delta_{x}, D_{\mu}(x) ).
\]
We denote their corresponding probability density functions by 
\[
 \phi_{n}( r \mid \delta_{x}, D_{\mu}(x) ).
\]

In practice, the sample probability density is often a better approximation to the true probability density of the distribution being sampled from than the sample distribution is of the distribution.
Thus we seek to approximate the probability density functions $\phi_{n}( r \mid \delta_{x}, D_{\mu}(x) )$
from the data given to us.
For a given $n$, 
we call the estimate this gives us for $D_{\mu}( x )$  the $n^{\text{th}}$ {\em nearest-neighbor dimension}.
\\

As we have seen, the $n^{\text{th}}$ nearest-neighbor dimensions provide good estimates to pointwise dimensions for measures satisfying the local uniformity condition (\ref{eq-LocalUniformity}).
We remark that there are certainly pathological measures for which this condition does not hold -- for example, the measures defined in Billingsley \cite{Billingsley1961}, which Cutler \cite{Cutler1993} terms ``fractal measures''.
It is an interesting problem to determine how one may give a limit-free description of pointwise dimension in the case of such measures.
\\

As a final point of note, it is the densities $\delta_{x}$ in this description which enables us to build into our algorithm the desired insensitivity to locally bi-Lipschitz injections of the data which we discussed in the previous section.

\subsection{Algorithm Overview \label{sec-AlgorithmOverview}}
Suppose that we are given a sequence
\[
 \mathcal{D} = \{ x_{1}, x_{2}, \hdots, x_{k} \} 
\] 
in $\mathbb{R}^{N}$ of data sampled from a Borel probablity measure $\mu$. 
For each $1 \le j \le k$ and each $1 \le i < k$, let us denote by $r_{j}( i )$ the distance from $x_{j}$ to its $i^{\text{th}}$ closest neighbor in $\mathcal{D}$.
For each $1 \le i < k$, we define the sequence 
\[
 \mathcal{R}_{i} := \{ r_{1}( i ), r_{2}( i ), \hdots, r_{k}( i ) \}.
\]

In our algorithm, a parameter $n$ is specified at the outset. Our goal will be to use the $n^{\text{th}}$-nearest neighbor distances $\mathcal{R}_{n}$ of the data $\mathcal{D}$ in order to obtain an approximation 
\begin{equation}
 \label{eq-AlgorithmOutput}
 \mu = \sum_{ m = 1 }^{ M } \theta_{m} \mu_{m},
\end{equation}
where 
\begin{enumerate}
 \item $0 < \theta_{m} \le 1$, $\sum_{j=1}^{M}\theta_{m} = 1$.
 \item Each $\mu_{m}$ is a Borel probability measure satisfying the local uniformity condition of (\ref{eq-LocalUniformity}).
 \item Each measure $\mu_{m}$ has constant pointwise dimension $d_{m} = D_{\mu_{m}}( x )$ and constant density $\delta_{m} = \delta_{x}$ for all points $x$ in its support.
 \item The measures $\mu_{m}$ have disjoint supports.
\end{enumerate}

Given (\ref{eq-AlgorithmOutput}), we can partition the data $\mathcal{D}$ into $M$ clusters defined by the supports of the measures $\mu_{m}$. $\theta_{m}$ represents the proportion of the data in the cluster $K_{m}$ corresponding to $\mu_{m}$.
\\

The quality of an approximation of the type specified by (\ref{eq-AlgorithmOutput}) depends on the choice of $M$.
For example, if $M=1$, then any dimension estimate derived from the approximation reflects an aggregate of the local scaling behaviors of $\mu$ around the data points. On the other hand, if $M=k$ (the size of the data $\mathcal{D}$), then one obtains a clear picture of the individual scaling behaviors, but at great cost. Often it is not necessary to go to this latter extreme in order to estimate the distribution of pointwise dimensions of $\mu$ at the data points.
An important feature of our algorithm is that it chooses the clusters and the number of clusters adaptively, balancing concerns of efficiency with those of clarity.
We will describe this in greater detail in Section \ref{sec-Clustering}.
\\

The bigger problem is actually producing the measures $\mu_{m}$.
In order to do this, we will use a Variational Bayesian method.
The idea is that, since each $\mu_{m}$ satisfies the local uniformity condition (\ref{eq-LocalUniformity}) with constant dimension and density, 
the $n^{\text{th}}$-nearest neighbor distribution for each cluster $K_{m}$ of the data should have 
a probability density function which is approximately $\phi_{n}$ from (\ref{eq-NearestNeighborDistribution})
 for the appropriate choice of parameters $d_{m}$ and $\delta_{m}$.
As a part of this procedure, we use the density functions $\phi_{n}$ to produce posterior distributions for the parameters $d_{m}$ and $\delta_{m}$ having made a choice of prior distributions.  
It is important to note that this forces us to treat the dimension $d_{m}$ (as well as the other parameters) implicitly as random variables. It is also worth noting that the choice of prior distributions for the relevant parameters can have a marked impact on the performance of the estimator when the data is sparse.
\\

The cluster parameters -- the number of clusters $M$, the cluster assignments $K_{m}$, and consequently the weights $\theta_{m}$ -- are estimated concurrently with the parameters $d_{m}$ and $\delta_{m}$.
In order to do this, 
we make use of the latent cluster assignment variables
\[
 Z_{im} := 
\left\{ 
 \begin{array}{ cl }
  1, & \text{if}\ x_{i} \in K_{m}, \\
  0, & \text{else}
 \end{array}\right..
\]

Most of the optimization in our algorithm is conditioned on the number of clusters $M$.
In these steps, the variables that we seek to estimate from the data are $d_{m}, \delta_{m}$, and $\theta_{m}$.
Our estimates for the variables $Z_{i,m}$ are derived from those. 
Thus we require prior distributions for $d_{m}, \delta_{m}$, and $\theta_{m}$.
We assume that the variables $d_{m}$ and $\delta_{m}$ follow independent gamma distribution with the parameter $(\alpha_{d}, \beta_{d})$ and $(\alpha_{\delta}, \beta_{\delta})$ respectively.
For each integer $ M > 0 $, we assume that the vector $( \theta_{1}, \theta_{2}, \hdots, \theta_{M})$ follows a Dirichlet distribution with the parameter $\left( \gamma_{1}^{(M)}, \gamma_{2}^{(M)}, \hdots, \gamma_{M}^{(M)} \right)$.
Figure \ref{fig-GraphicalModel} provides a graphical model depicting the dependency between these variables.
\\

 \begin{figure}[htb, clip]
   \begin{center}
    \includegraphics[width=1\linewidth, clip]{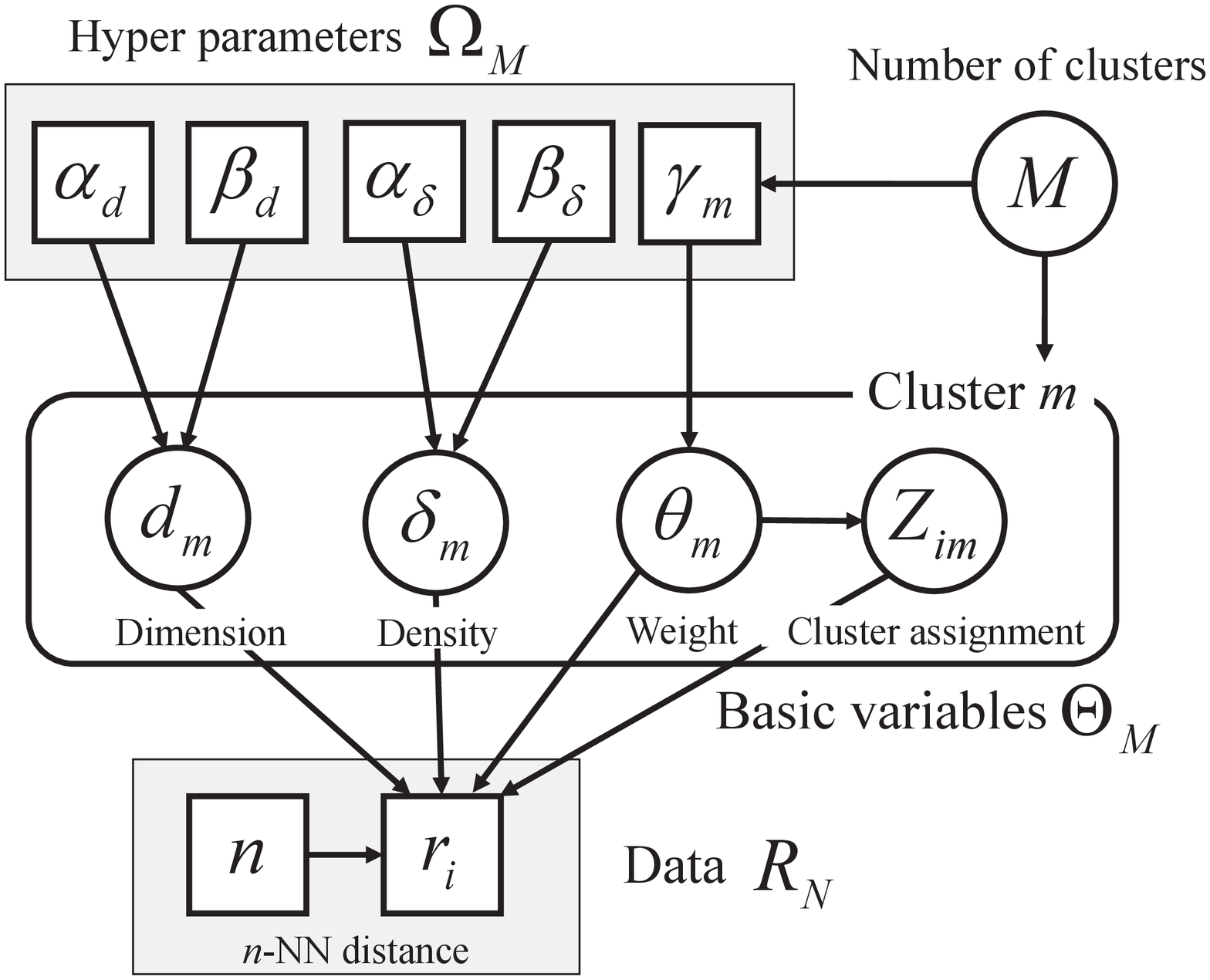}
    \caption{ Dependencies among the variables.
\label{fig-GraphicalModel}}      
    \end{center}
  \end{figure}


\subsubsection{Glossary \label{sec-Glossary}}
{\bf Data:} 
\begin{itemize}
 \item $\mathcal{D}$ --- a sequence $x_{1}, x_{2}, \hdots, x_{k}$ of $k$ data points in $\mathbb{R}^{N}$ sampled from the Borel probability measure $\mu$ for which we are trying to estimate the distribution of pointwise dimensions.
 \item $\mathcal{R}_{n}$ --- the sequence of distances $r_{1}, r_{2}, \hdots, r_{k}$ from each of the data points in $\mathcal{D}$ to its $n^{\text{th}}$-nearest neighbor.
\end{itemize}

{\bf Basic variables used in estimation:} 
 \begin{itemize}
 \item $M$ --- the number of clusters in the approximation (\ref{eq-AlgorithmOutput}).
 \item $\theta_{m}$  --- the weight of cluster $m$ in (\ref{eq-AlgorithmOutput}).
 \item $d_{m}$ --- the pointwise dimension of $\mu_{m}$ in (\ref{eq-AlgorithmOutput}).
 \item $\delta_{m}$  --- the density parameter in the Poisson process approximating $\mu_{m}$ according to the discussion in Section \ref{sec-LimitFree}.
 \item $Z_{im}$ --- the indicator that the $i^{\text{th}}$ data point in $\mathcal{D}$ belongs to cluster $m$.
 \end{itemize}

{\bf Hyper-parameters:} 
 \begin{itemize}
 \item $( \alpha_{d}, \beta_{d} )$ --- the parameters for the gamma prior of the dimension parameter $d_{m}$.
  \item $( \alpha_{\delta}, \beta_{\delta} )$ --- the parameters for the gamma prior of the density parameter $\delta_{m}$.
  \item $\gamma^{(M)} := ( \gamma_{1}^{(M)}, \hdots, \gamma_{M}^{(M)} )$ --- the parameters for the Dirichlet prior of the weight parameters $\theta^{(M)} := (\theta_{1}, \hdots, \theta_{M})$.
 \end{itemize}

{\bf Some short-hand:}
 \begin{itemize}
 \item $\Theta_{M} = \left( ( Z_{im} ), ( d_{m} ), ( \delta_{m} ), ( \theta_{m} ) \right)$ --- a list of the basic variables apart from $M$.
 \item $\Theta = \left( M, \Theta_{M} \right)$ --- the list of all the basic variables.
 \item $\Omega_{M} = \left( ( \alpha_{d}, \beta_{d} ), ( \alpha_{\delta}, \beta_{\delta} ), ( \gamma_{1}^{(M)}, \hdots, \gamma_{M}^{(M)} )\right)$ --- the list of hyper-parameters for a fixed number $M$ of clusters.
 \item $\Omega = \left( M, \Omega_{M} \right)$ --- here, the number of clusters is allowed to vary.
\end{itemize}

\subsubsection{The Objective}
The goal of the procedure described above is to compute the posterior distribution
for the list of variables $\Theta$ given $\Omega$ and $\mathcal{R}_{N}$:
\begin{equation}
 P\left( \Theta \mid \Omega, \mathcal{R}_{N} \right) = 
\frac{P\left( \mathcal{R}_{N} \mid \Theta \right) P\left( \Theta \mid \Omega \right)}
{ P\left( \mathcal{R}^{N} \mid \Omega \right)}.
 \label{eq-PosteriorDistribution}
\end{equation}

\subsection{The Variational Bayesian Algorithm \label{sec-VBApproximation}}
\subsubsection{Phases \label{sec-Phases}}
As we have a clean description (\ref{eq-NearestNeighborDistribution}) of the relationship between the basic variables, we can feasibly use the variational Bayesian method presented by Jordan et al. \cite{Jordan1997,Attias1999,Beal2003,Ghahramani1999} to estimate the posterior distribution (\ref{eq-PosteriorDistribution}).
There are two distinct phases in our application of this method: 
\begin{itemize}
 \item {\bf The fixed phase} --- here, we estimate the posterior distribution conditioned upon the number of clusters $M$.
 \item {\bf The clustering phase} --- here, we use the techniques of Ueda et al. \cite{Ueda2000,Ueda2002} to find the number of clusters $M$ for which the estimate from the corresponding fixed phase best approximates the posterior of (\ref{eq-PosteriorDistribution}).
\end{itemize}
Although we treat these phases separately, there is considerable between them as we iteratively update our estimate of the posterior (\ref{eq-PosteriorDistribution}).
\\

In the following section, we describe the connection between the two phases of our procedure -- the objective function that our estimator seeks to maximize.

\subsubsection{The Objective Function \label{sec-ObjectiveFunction}}
At each step, and for each component $X$ of $\Theta$, we distinguish between the true marginal posterior distribution $P( X | \Omega, \mathcal{R}_{N} )$ derived from (\ref{eq-PosteriorDistribution})
 and our estimates of this distribution $Q( X | \Omega )$, which we also denote $Q( X )$ for convenience as we are not varying the hyper-parameters $\Omega$.
Finally, this is crucial, we assume that these marginal distributions of $Q$ are independent for each fixed number of clusters $M$: 
\begin{equation}
 \label{eq-VariationalPosteriorIndependence}
  Q( \Theta_{M} ) = \prod_{i=1}^{k} Q( \theta^{(M)} ) \prod_{m=1}^{M}
Q( Z_{im} )  Q( d_{m} ) Q( \delta_{m} ). 
\end{equation}

We can estimate the quality of the approximation $Q$ by considering the likelihood $P( \mathcal{R}_{N} \mid \Omega )$ which measures how well the data is determined by the hyper-parameters $\Omega$. 
To begin with, 
\[
  \log P( \mathcal{R}_{N} | \Omega_{M} ) = \log \int P( \mathcal{R}_{N}, \Theta_{M} \mid \Omega_{M} ) \D \Theta_{M}.
\]
Let us write 
\begin{equation}
 \label{eq-ObjectiveFunction}
  L_{M}( Q ) = \int Q\left( \Theta_{M} \right) \log \frac{ P\left( \mathcal{R}_{N}, \Theta_{M} \mid \Omega_{M} \right) }{ Q\left(  \Theta_{M} \right) } 
 \D \Theta,
\end{equation}
and let us denote by $\kappa_{M}( Q )$ the Kullback-Leibler divergence 
\begin{equation}
 \label{eq-KLDivergenceM}
 \kappa_{M}( Q ) := KL\left( Q( \Theta_{M} ) \| P( \Theta_{M}, \mathcal{R}_{N} \mid \Omega_{M} ) \right).
\end{equation}
We have the decomposition
\begin{equation}
 \label{eq-KLError}
 \log P( \mathcal{R}_{N} | \Omega_{M} ) = 
L_{M}( Q ) + \kappa_{M}( Q ).
\end{equation}
Our goal, in the fixed phase of the procedure, is to find the distribution $Q$ which maximizes the objective function $L_{M}( Q )$. This is equivalent from the above decomposition to minimizing the Kullback-Leibler divergence $\kappa_{M}( Q )$, which measures the error in our approximation.
\\

In our implementation, we use the explicit description of the objective functions $L_{M}( Q )$ to perform the optimization. 
We will describe the details for each of the phases separately once we have presented explicit descriptions of the relevant prior distributions, which the following section is devoted to.


\subsubsection{The Prior Distributions}

In this section, we present expressions for the prior distributions  $P( X \mid \Omega_{M})$ which we will make use of in our implementation. We also describe the distribution of the data conditioned on the basic variables $P( \mathcal{R}_{N} \mid \Theta_{M} )$ which we use in conjunction with the priors to derive the posterior distribution $P( \Theta_{M} \mid \Omega_{M}, \mathcal{R}_{N} )$.
\\

We begin by describing our assumed prior distributions for the basic parameters given a fixed number of clusters $M$. 
We stated our assumptions about these parameters at the end of Section \ref{sec-AlgorithmOverview}, but we provide here explicit expressions of their probability densities: 
\begin{enumerate}
 \item {\bf The density parameters ---} The density parameter $\delta_{m}$ is assumed to follow a gamma distribution with the shape parameter $\alpha_{\delta} > 0$ and the rate parameter $\beta_{\delta} > 0$:
\begin{equation}
 \label{eq-PriorDensity}
 P( \delta_{m} | \alpha_{\delta}, \beta_{\delta} ) 
= 
\frac{ \beta_{\delta}^{\alpha_{\delta} } \delta_{m}^{ \alpha_{\delta} - 1 } }{ \Gamma( \alpha_{\delta} ) }
\exp\left( - \beta_{\delta} \delta_{m} \right).
\end{equation}
 \item {\bf The dimension parameters ---} The dimension parameter $d_{m}$ is assumed to follow a gamma distribution $\alpha_{d} > 0$ and $\beta_{d} > 0$:
\begin{equation}
 \label{eq-PriorDimension}
 P( d_{m} | \alpha_{ d }, \beta_{ d } ) 
= 
\frac{ \beta_{ d }^{\alpha_{ d } }  d_{m}^{ \alpha_{ d } - 1 } }{ \Gamma( \alpha_{ d } ) }
\exp\left( - \beta_{ d }  d_{m} \right).
\end{equation}
 \item {\bf The weight parameters ---} The weight parameter $\theta := \theta^{(M)}$ follows a Dirichlet distribution with the exponents $\gamma := \gamma^{(M)}$:
\begin{equation}
 \label{eq-PriorWeight}
 P\left( \theta | \gamma \right) 
= 
\frac{ \Gamma\left( \sum_{ m = 1 }^{M} \gamma_{ m } \right) }{ \prod_{ m = 1 }^{M}\Gamma( \gamma_{ m } ) }
\prod_{m=1}^{M} \theta_{m}^{\gamma_{m} - 1}.
\end{equation}
 \item {\bf The cluster indicators ---} The cluster indicators $Z_{i} := ( Z_{i1}, Z_{i2}, \hdots, Z_{iM} )$ follow a multinomial distribution with the probabilities $\theta$: 
\begin{equation}
 P( Z_{i} \mid \theta ) = \prod_{m=1}^{M}\theta_{m}^{Z_{im}}
\end{equation}
\end{enumerate}

Finally, the conditional probability of drawing the samples $\mathcal{R}_{N}$ given a set of basic variables $\Theta_{M}$ is
\begin{equation}
  P( \mathcal{R}_{N} \mid \Theta_{M}) = \prod_{i = 1}^{N}\prod_{m=1}^{M}\left\{ \theta_{m}^{(M)}P( r_{i} \mid n, \delta_{m}, d_{m} ) \right\}^{Z_{im}}
\label{eq-LikelihoodAppendix}
\end{equation}
where 
\begin{equation}
 P( r_{i} \mid n, \delta_{m}, d_{m} ) = \frac{ d_{m} \delta_{m}^{n} r_{i}^{ n d_{m} - 1 } }{ \Gamma(n) }\exp\left( - \delta r_{i}^{d_{m}} \right).
\label{eq-FormalNNDistributionClusters}
\end{equation} 
Note that (\ref{eq-LikelihoodAppendix}) and (\ref{eq-FormalNNDistributionClusters}) are derived from (\ref{eq-NearestNeighborDistribution}), accounting for differences in the clusters.
\\

The above expressions allow us to explicitly describe the objective function (\ref{eq-ObjectiveFunction}). 
This description is provided in (\ref{eq-KLDivergenceDetail}), and is developed in the following section.
\begin{figure*}
\begin{eqnarray}
\label{eq-KLDivergenceDetail}
L_{M}( Q ) 
&=& \sum_{i}^{k}\sum_{m=1}^{M} \expf{ t }{ Z_{im} }
\left\{
-\log\left( \expf{ t }{ Z_{im} } \right) + \expf{t}{ \log\left( \theta_{m} \right) }
+ \expf{ t }{ \log( d_{m} ) } 
+ n \expf{ t }{ d_{m} } \log( r_{i} ) + n \expf{ t }{ \log( \delta_{m} ) }
-\expf{ t }{ \delta_{m} } \expf{ t }{ r_{i}^{ d_{m} } }
\right\}
\nonumber
\\
&-& k \log \Gamma( n ) - \sum_{m=1}^{M} \log \Gamma\left( \gamma_{m} \right) + \log \Gamma\left( \sum_{m=1}^{M}\gamma_{m} \right) 
+ M \log\left( \frac{ \beta_{\delta}^{\alpha_{\delta}} }{ \Gamma( \alpha_{\delta} ) } \right)
+ M \log\left( \frac{ \beta_{d}^{\alpha_{d}} }{ \Gamma( \alpha_{d} ) } \right)
\\
&+& \sum_{m=1}^{M}\log \Gamma\left( \hat{\gamma}_{m}^{(t)} \right)
- \log \left( \sum_{m=1}^{M} \Gamma\left( \hat{\gamma}_{m}^{(t)} \right) \right)
- 
 \sum_{m=1}^{M}
\left\{
\hat{\alpha}_{m}^{(t)} \log\left( \hat{ \beta }_{m}^{(t)} \right)
- \log\left( \Gamma\left( \hat{ \alpha }_{m}^{(t)} \right) \right) - \hat{ \beta }_{m}^{(t)}\expf{ t }{ \delta_{m} }
\right\} 
 + \sum_{m=1}^{M} \hat{D}^{(t)}_{m}
\nonumber
\end{eqnarray}
\end{figure*}
%
%
%

\subsubsection{The Fixed Phase \label{sec-FixedPhase}}
Recall that in the fixed phase of our algorithm we fix a number of clusters $M$ and attempt to find an approximate posterior distribution $Q( X )$ for each component $X$ of $\Theta_{M}$.
This approximation is obtained by iteratively maximizing the objective function  $L_{M}(Q)$.
In this section, we describe what happens in a single iteration of this optimization procedure.
We denote by $Q_{t}$ the approximation obtained at the $t^{\text{th}}$ step, and 
we also introduce the notation $\Theta_{M, \neg X}$ for those components of $\Theta_{M}$ which are distinct from $X$.
\\

Each step of the optimization is further divided into four portions. These portions correspond to the four classes of basic parameters distinct from the number of clusters $M$ -- the density parameters $\delta_{m}$, the dimension parameters $d_{m}$, the weight parameters $\theta = \theta^{(M)}$, and the cluster indicators $Z_{i} = ( Z_{i1}, Z_{i2}, \hdots, Z_{iM} )$.
In each portion, we modify the values of the relevant parameters conditioned upon those of the other classes of parameters.
The goal at each step is to increase the value of the objective function $L_{M}(Q)$, so that we always have 
\[
 L_{M}( Q_{t} ) \le L_{M}( Q_{t+1} ).
\]
The convergence of such methods has been studied extensively by Attias \cite{Attias1999}, and by Ghahramani and Beal \cite{Ghahramani1999}, although it has not been rigorously proven.
\\

Stated plainly, our goal is to approximate the true distributions of the basic variables $\delta_{m}$, $d_{m}$, $\theta$, and $Z_{i}$ from within a constrained family of posterior distributions.
In our approximation, we will assume that the densities $\delta_{m}$ follow a gamma distribution, that the weights $\theta$ follow a Dirichlet distribution, and that the cluster indicators $Z_{i}$ follow a multinomial distribution.
The posteriors for the dimensions will be derived from our estimated posteriors of these variables.
Thus, at step $t$, the situation is as follows: 
\begin{enumerate}
 \item The estimated posterior distribution $Q_{t}( \delta_{m} )$ is a gamma distribution with the estimated parameters $\hat{\alpha}_{m}^{(t)}$ and $\hat{\beta}_{m}^{(t)}$.
 \item The estimated posterior distribution $Q_{t}( \theta )$ is a Dirichlet distribution with the estimated parameters $\hat{\gamma}^{(t)} = \left( \hat{\gamma}_{1}^{(t)}, \hat{\gamma}_{2}^{(t)}, \hdots, \hat{\gamma}_{M}^{(t)}\right)$.
 \item The estimated posterior distribution $Q_{t}( Z_{i} )$ is a multinomial distribution with the estimated parameters $\hat{\pi}_{i}^{(t)} := \left( \hat{\pi}_{i1}^{(t)}, \hat{\pi}_{i2}^{(t)}, \hdots, \hat{\pi}_{iM}^{(t)}\right)$.
 \item The estimated posterior distribution $Q_{t}( d_{m} )$ is derived from the above posterior distributions in a manner that we shall describe below. 
This distribution has the parameters $\hat{A}_{m}^{(t)}$, $\hat{B}_{m}^{(t)}$, and $\hat{C}_{im}^{(t)}$ for $1 \le i \le N$, and its probability density function is 
\begin{equation}
\label{eq-PosteriorDimension}
 Q_{t}( d_{m} ) = \frac{ d_{m}^{ \hat{A}_{m}^{(t)} } }{ \hat{D}_{m}^{(t)} } \exp\left( 
 d_{m} \hat{B}_{m}^{(t)} - \sum_{i=1}^{N} \hat{C}_{im}^{(t)} r_{i}^{d_{m}}
\right) 
\end{equation}
where the denominator $\hat{D}_{m}^{(t)}$ normalizes the integral of (\ref{eq-PosteriorDimension}).
\end{enumerate}

We will need another bit of notation before we can describe the rules by which we update the posterior distributions.
For a component $X$ of $\Theta_{M}$ and a function $f$  of $X$, we write
\[
 \expf{t}{f(X)} := \expectation{Q_{t}(X) }{ f( X ) },
\]
the expectation of $f( X )$ with respect to the variational posterior distribution $Q_{t}( X )$.
\\

We now present the update rules. 
These rules are derived from the procedure presented in Jordan et al. \cite{Jordan1997}, and we refer the reader to that paper for details.
According to \cite{Jordan1997}, 
for each component $X$ of $\Theta_{M*}$, the updating rule for the variational posterior is written with (\ref{eq-LikelihoodAppendix}):
\begin{equation}
 \label{eq-Maximization}
 Q_{t+1}( X ) \propto P( X | \Omega_{M} ) e^{ \expectation{Q_{t}( \Theta_{M, \neg X } )}{ \log\left( P( \mathcal{R}_{n} | \Theta_{M} )  \right) } }
\end{equation}
Pseudo-code corresponding to each rule may be found in Section \ref{sec-PseudoCode}.
In parentheses, we name each updating rule for use in our pseudo-code in Section \ref{sec-PseudoCode}.
\\

\noindent \textbf{Update rules for $\hat{\alpha}_{m}^{(t)}$ and $\hat{\beta}_{m}^{(t)}$ }\\ 
(\textbf{PosteriorDensity } in Section \ref{sec-PseudoCode})
\[
 \hat{\alpha}_{m}^{(t+1)} := \alpha_{\delta} + n \sum_{i=1}^{k}\expf{t}{ Z_{im} },
\]
and
\[
 \hat{\beta}_{m}^{(t+1)} := \beta_{\delta} + \sum_{i=1}^{k} \expf{t}{ Z_{im} }\expf{ t }{ r_{i}^{d_{m}} }.
\]

\noindent \textbf{Update rules for $\hat{\gamma}^{(t)}$ }\\ 
(\textbf{PosteriorWeight } in Section \ref{sec-PseudoCode})
\begin{equation}
 \hat{\gamma}_{m}^{(t+1)} := \gamma_{m}^{(M)} +  \sum_{i=1}^{k}\expf{t}{ Z_{im} }.
\label{eq-PosteriorFrequency}
\end{equation}

\noindent \textbf{Update rules for $\hat{\pi}^{(t)}$ }\\ 
(\textbf{PosteriorCluster } in Section \ref{sec-PseudoCode})
\begin{eqnarray}
\label{eq-PosteriorCluster}
\log\left( \hat{\pi}_{m}^{(t+1)} \right) &\propto&
\expf{ t }{ \log\left( \theta_{m} \right) } 
+
\expf{ t }{ \log\left( d_{m} \right) } 
\\
&+& 
n\expf{ t }{ d_{m} } \log( r_{i} )
+
n\expf{ t }{ \log\left( \delta_{m} \right) }
\\
&+&
\expf{ t }{ \delta_{m} } \expf{ t }{ \log\left( d_{m} \right) } .
\end{eqnarray}
Here, we have
\[
 \expf{t}{ \delta_{m} } = \frac{ \hat{ \alpha }_{m}^{(t)} }{ \hat{ \beta }_{m}^{(t)} },
\]
\[
 \expf{ t }{ \log\left( \theta_{m} \right) }  
 = \Psi\left( \hat{\gamma}_{m}^{(t)} \right) - \Psi\left( \sum_{m=1}^{M}\hat{\gamma}_{m}^{(t)} \right), 
\]
and 
\[
 \expf{ t }{ \log\left( \delta_{m} \right) }  = 
 \Psi\left( \hat{ \alpha }_{m}^{(t)} \right) - \log \hat{ \beta }_{m}^{(t)},
\]
where $\Psi( x ) = \frac{ \D  }{ \D x } \log( \Gamma( x ) ) $ is the digamma function.
The other expectations do not have such simple expressions and are estimated numerically.
\\

\noindent \textbf{Update rules for $\hat{A}_{m}^{(t)}$, $\hat{B}_{m}^{(t)}$, and $\hat{C}_{im}^{(t)}$ }\\ 
(\textbf{PosteriorDimension } in Section \ref{sec-PseudoCode})
%
Updating rule of the dimension parameter $d_{m}$ ({\bf (PosteriorDimension)} in Section \ref{sec-PseudoCode}).
\[
 \hat{A}_{m}^{(t+1)} = \sum_{i=1}^{k} \expf{t}{ Z_{im} } + \alpha_{d} - 1,
\]
\[
 \hat{B}_{m}^{ (t+1) } = n  \sum_{i=1}^{N} \expf{t}{ Z_{im} } \log r_{i} + \beta_{d},
\]
and
\[
  \hat{C}_{im}^{(t+1)} = \expf{t}{\delta_{m}} \expf{t}{ Z_{im} }.
\]
The denominator $\hat{D}_{m}^{(t+1)}$ of (\ref{eq-PosteriorDimension}) is numerically integrated.
This denominator plays an important role in the evaluation of the objective function $L_{M}(Q)$ as seen in (\ref{eq-KLDivergenceDetail}).
Its numerical estimation is a bottleneck in our algorithm.
\\

This concludes our discussion of the fixed phase of our algorithm. We now discuss the choice of the number of clusters $M$.

\subsubsection{The Clustering Phase \label{sec-Clustering}}

In this section, we discuss how we adaptively choose the number of clusters $M$ in the approximation (\ref{eq-AlgorithmOutput}) using the split-and-merge technique of Ueda et al. \cite{Ueda2000,Ueda2002}.
We describe here the essence of this method and refer the reader to those papers for details.
\\

For a given approximation of the form (\ref{eq-AlgorithmOutput}), 
there are many ways in which one could potentially split or merge clusters.
In the current implementation of our algorithm, we have essentially chosen to attempt splits and merges at random.
There is a quite bit of scope for improvement in this department.
However, for the rest of this section, let us assume that we have either chosen a cluster to split or chosen two clusters to merge together without concerning ourselves with how the choice was made.
Thus, from the original list of $M$ clusters, we are considering whether or not to use $M^{*}$ clusters where $M^{*}$ is either $M-1$ or $M+1$. 
\\

In either case, we begin by adjusting the hyper-parameter $\gamma^{(M)}$.
In the case of a proposed split, $\gamma^{(M^{*})}$ is derived from $\gamma^{(M)}$ by distributing the component corresponding to the cluster to be split evenly among the two new clusters. 
In the case of a proposed merge, the component corresponding to the cluster to be merged are simply added together to produce that of the new cluster.
\\

The next step is to assign values to the parameters for the variational posteriors of the densities and dimensions which we described in the previous section.
For most of the clusters, these parameters will remain unchanged in the new model.
In the case where we are considering a split, one of the new clusters will have the same parameter values as the original cluster whereas the second new cluster will have parameter values which are perturbations of these values.
In the case where we are considering a merge, 
the parameter values for the new cluster is the average of those for the original clusters.
Once these values have been assigned, we use the updating rules presented in the previous section to derive variational posterior distributions $Q^{*}$ for the basic variables conditioned on the number of clusters being $M^{*}$.
This allows us to compare the values of the objective functions $L_{M}( Q )$ and $L_{M^{*}}( Q^{*} )$ using (\ref{eq-ObjectiveFunction}).
We accept the proposed split or merge if
\[
  L_{M^{*}}( Q^{*} ) > L_{M}( Q ).
\]

Further technical details of our implementation are clear from the pseudo-code in Section \ref{sec-PseudoCode}, where the clustering phase is presented in the form of the functions {\bf Split} and {\bf Merge}.

\subsection{Implementation \label{sec-PseudoCode}}
The first problem we address in our implementation is that of initializing the basic variables.
We begin by assuming that there is only one cluster -- $M=1$.
Under this assumption, we use the maximum likelihood estimates of the basic variables corresponding to the conditional distribution $P(\mathcal{R}_{n} | \Theta_{1} )$ derived in (\ref{eq-LikelihoodAppendix}).
The basic variables can also be initialized under the assumption that the initial number of clusters is some fixed number larger than 1. In this situation, we use 
the EM algorithm of Dempster et al. \cite{Dempster1977}.
\\

The next issue is of deciding when the optimization procedure should halt.
Our implementation takes an argument $t_{\max}$, which specifies the maximal number of iterations to perform. 
It also takes an argument $\Delta L$, which specifies the minimum acceptable improvement in the value of the objective function.
If we denote by $L_{t}$ the value of the objective function at step $t$, the program halts if  $t = t_{\max}$ or $L_{t} - L_{t-1} < \Delta L$.
\\

The pseudo-code below details our implementation under these conditions.
The {\bf Main} function takes the data, the number $n$ of nearest neighbors to be used, and the initial number of clusters as arguments.
It calls each of the functions {\bf Update}, {\bf Split}, and {\bf Merge} iteratively, checking for the halting criteria at each step.
The function {\bf Update} implements the fixed phase of our algorithm, and the clustering phase is handled by {\bf Split} and {\bf Merge}.

The {\bf Objective} function evaluates the objective function $L_{M}( Q )$ for a given list of basic variables.
The functions whose names begin with `{\bf Posterior}' are implementation of the update rules which may be found in Section \ref{sec-FixedPhase}.
The {\bf LogNormal} function returns a sample from the log-normal distribution with the specified mean and variance parameters.
The function $\mathbf{Random}( \{ 1, \hdots, M \}, k ) $ randomly chooses a subset of $\{1, 2, \hdots, M\}$ of size $k$.


$\ $

\begin{algorithmic}[1] 
\STATE $\mathbf{Main}( \mathcal{D}, n, M, \Delta L, t_{\text{max}} )$
\STATE $\mathcal{R}_{n} \leftarrow \mathbf{NearestNeighborDistances}( \mathcal{D}, n )$
\STATE $\{ d_{m}, \delta_{m}, \theta_{m}, Z_{im} \} \leftarrow \mathbf{Initialization}( \mathcal{R}_{n}, M )$
\STATE $L_{0} \leftarrow \inf$
\STATE $L_{1} \leftarrow 0$
\STATE $\Theta_{M} \leftarrow ( \{ { d }_{m}, { \delta }_{m}, { \theta }_{m}, { Z }_{m} \}_{m=1}^{M} )$
\WHILE { $| L_{t} - L_{t+1} | > \Delta L$ \AND $t \le t_{\text{max}}$ }
 \STATE $t \leftarrow t + 1$
 \STATE $( L_{t}, \Theta_{M} ) \leftarrow \mathbf{Update}( \Theta_{M} )$
 \STATE $( \bar{ L }_{t}, \bar{ \Omega }_{M-1} ) \leftarrow \mathbf{Merge}( \Theta_{M} \}_{m=1}^{M} )$
 \IF{ $\bar{ L }_{t} > L_{t}$ }
 \STATE $L_{t} \leftarrow \bar{L}_{t}$, $\Theta_{M-1} \leftarrow \bar{ \Omega }_{M-1}$, $M \leftarrow M - 1 $
 \ENDIF
 \STATE $( \bar{L}_{t}, \bar{\Omega}_{M+1} ) \leftarrow \mathbf{Split}( \Theta_{M} )$
 \IF{ $\bar{ L }_{t} > L_{t}$ }
 \STATE $L_{t} \leftarrow \bar{L}_{t}$, $\Theta_{M+1} \leftarrow \bar{ \Omega }_{M+1}$, $M \leftarrow M + 1 $
 \ENDIF
\ENDWHILE
\RETURN $( \Theta_{M} )$
\end{algorithmic}

$\ $

\begin{algorithmic}[1]
\STATE $\mathbf{Update}( \Theta_{M} = \{ d_{m}, \delta_{m}, \theta_{m}, Z_{im} \}_{m = 1}^{M} )$
\STATE Initialize: $t = 0$, $L_{0} \leftarrow -\infty$, $L_{1} \leftarrow 0$
\WHILE {$| L_{t} - L_{t+1} | > \Delta L$}
 \STATE $t \leftarrow t + 1$
 \STATE $\{d_{m}\} \leftarrow \mathbf{PosteriorDimension}( \{\delta_{m}, \theta_{m}, Z_{im}\}_{m=1}^{M} )$
 \STATE $\{\delta_{m}\} \leftarrow \mathbf{PosteriorDensity}( \{d_{m}, \theta_{m}, Z_{im}\}_{m=1}^{M} )$
 \STATE $\{\theta_{m}\} \leftarrow \mathbf{PosteriorWeight}( \{d_{m}, \delta_{m}, Z_{im}\}_{m=1}^{M} )$
 \STATE $\{Z_{im}\} \leftarrow \mathbf{PosteriorCluster}( \{d_{m}, \delta_{m}, \theta_{m}\}_{m=1}^{M} )$
 \STATE $L_{t} \leftarrow \mathbf{Objective}( \{d_{m}, \delta_{m}, \theta_{m}, Z_{im} \}_{m=1}^{M} )$
\ENDWHILE
\STATE $\Theta_{M} \leftarrow \{ d_{m}, \delta_{m}, \theta_{m}, Z_{im} \}_{m = 1}^{M}$
\RETURN $( L_{t}, \Theta_{M} )$
\end{algorithmic}

$\ $ 

\begin{algorithmic}[1]
\STATE $\mathbf{Split}( \Theta_{M} = \{d_{m}, \delta_{m}, \theta_{m}, Z_{im}\}_{m=1}^{M} )$
\STATE $K \leftarrow \mathbf{Random}( \{ 1, \hdots, M \}, 1 ) $
\STATE $\Theta_{K} \leftarrow \{ d_{m}, \delta_{m}, \theta_{m}, Z_{im} \}_{m \in K}$
\STATE $(L, \Theta_{2}) \leftarrow \mathbf{Update}( \mathbf{InitializeSplit}( \Theta_{K} ) )$
\STATE $\Theta_{M+1} \leftarrow \{ \Theta_{2}, \Theta_{M} \setminus \Theta_{K} \}$
\RETURN $(L, \Theta_{M+1})$
\end{algorithmic}

$\ $

\begin{algorithmic}[1]
\STATE $\mathbf{Merge}( \Theta_{M} = \{d_{m}, \delta_{m}, \theta_{m}, Z_{im}\}_{m=1}^{M} )$
\IF {$M = 1$}
\RETURN $( -\infty, \Theta_{M} )$
\ENDIF
\STATE $K \leftarrow \mathbf{Random}( \{ 1, \hdots, M \}, 2 ) $
\STATE $\Theta_{K} \leftarrow \{ d_{m}, \delta_{m}, \theta_{m}, Z_{im} \}_{m \in K}$
\STATE $(L, \Theta_{1}) \leftarrow \mathbf{Update}( \mathbf{InitializeMerge}( \Theta_{K} ) )$
\STATE $\Theta_{M-1} \leftarrow \{ \Theta_{1}, \Theta_{M} \setminus \Theta_{K} \}$
\RETURN $(L, \Theta_{M-1})$
\end{algorithmic}

$\ $

\begin{algorithmic}[1]
\STATE $\mathbf{NearestNeighborDistances}( \mathcal{D}, n )$
 \STATE Use the KD-tree algorithm to compute the list of $n^{\text{th}}$-nearest neighbor distances $\mathcal{R}_{n}$ of the data points in $\mathcal{D}$.
\RETURN $\mathcal{R}_{n}$
\end{algorithmic}

$\ $

\begin{algorithmic}[1]
\STATE $\mathbf{Initialization}( \mathcal{R}_{n}, M )$
 \STATE Use the EM algorithm to estimate the initial parameters $\{ \bar{ d }_{m}, \bar{ \delta }_{m}, \bar{ \theta }_{m}, \bar{ Z }_{im} \}$ which maximize the likelihood $P( \mathcal{R}_{n} | \Theta_{M} )$.
\RETURN $\{ \bar{ d }_{m}, \bar{ \delta }_{m}, \bar{ \theta }_{m}, \bar{ Z }_{im} \}$
\end{algorithmic}

$\ $

\begin{algorithmic}[1]
\STATE $\mathbf{InitializeSplit}( \Theta_{0} = \{d_{0}, \delta_{0}, \theta_{0}, Z_{i0}\} )$
\FOR{ $m = 1, 2$}
\STATE $d_{m} \leftarrow d_{0} \times \mathbf{LogNormal}(0, \sigma)$
\STATE $\delta_{m} \leftarrow \delta_{0} \times \mathbf{LogNormal}(0, \sigma)$
\STATE $\theta_{m} \leftarrow \theta_{0} \times \mathbf{LogNormal}(0, \sigma)$
\STATE $Z_{im} \leftarrow \mathbf{PosteriorCluster}( \{d_{m}, \delta_{m}, \theta_{m}\} )$
\ENDFOR
\STATE $\Theta_{2} = \{d_{m}, \delta_{m}, \theta_{m}, Z_{im}\}_{m=1}^{2}$
\RETURN $(\Theta_{2})$
\end{algorithmic}

$\ $

\begin{algorithmic}[1]
\STATE $\mathbf{InitializeMerge}( \Theta_{M} = \{d_{m}, \delta_{m}, \theta_{m}, Z_{im}\}_{m=1}^{M} )$
\STATE $d_{0} \leftarrow \frac{1}{M}\sum_{m=1}^{M}d_{m}$, $\delta_{0} \leftarrow \frac{1}{M}\sum_{m=1}^{M}\delta_{m}$
\STATE $\theta_{0} \leftarrow \sum_{m=1}^{M}\theta_{m}$, $Z_{i0} \leftarrow \sum_{m=1}^{M}Z_{im}$
\STATE $\Theta_{0} = \{d_{0}, \delta_{0}, \theta_{0}, Z_{i0}\}$
\RETURN $(\Theta_{0})$
\end{algorithmic}



\section{Basic Analysis \label{sec-BasicAnalysis}}

\subsection{Overview \label{sec-BasicAnalysisOverview}}

In this section, we test our estimator on data sampled from certain probability measures related to the Cantor measure.
Recall that we performed a similar analysis of GP estimators in Section \ref{sec-GPAnalysis}, and that we identified two serious issues as a result: 
\begin{enumerate}
 \item The difficulty in choosing informative scales.
 \item High sensitivity to transformations of the data under which dimension is invariant.
\end{enumerate}

There is one further thing that we observed -- 
not only does the difficulty in choosing informative scales reflect the dimension blindness of correlation dimension, but we also observed this same problem for data sampled from the Cantor measure which has uniform pointwise dimension.
\\

In this section, we will analyze the behavior of our proposed estimator when confronted with various data sets that exhibit behavior which was problematic for the GP estimator we analyzed in the manners described above. 
While the problem of having to choose informative scales has been addressed in the design of our estimator itself, the causes of the other problems were much subtler and call for a detailed analysis.
Before we perform our analysis, we formally describe the corresponding measures.

\subsection{Test Data \label{sec-BasicAnalysisData}}
Each of our test data sets is generated for the purpose of studying the effects on our estimator of one of the two following problems:
\begin{enumerate}
 \item Insensitivity to distribution of pointwise dimensions.
 \item Sensitivity to locally bi-Lipschitz injections.
\end{enumerate}
We present the data sets which test each of these issues in turn.

\subsubsection{dimension blindness \label{sec-BasicAnalysisDimensionBlindness}}
Our first class of measures is derived from the Cantor measure $\xi$ (which we described in detail in Section \ref{sec-PointwiseDimension}).
We call these measures {\em Cantor mixtures}, and they are essentially convex combinations of products of the Cantor measure with itself. 
The Cantor mixtures allow us to test to the sensitivity of our estimator to the distribution of pointwise dimensions of a generating measure.
The goal is to see how susceptible our estimator is to the problems caused by dimension blindness in the case of GP estimators -- the short answer is, ``not very''.
\\

We now describe this class explicitly: 

Choose positive integers $M$, $ N_{1}, N_{2}, \hdots, N_{M}$, and $N$. Choose vectors $v^{(1)}, v^{(2)}, \ldots, v^{(M)} \in \mathbb{R}^N$.
For $v \in \mathbb{R}^{N}$ and $1 \leq k \leq N$, let us write 
\[
 \xi_{v,k} := \left( \prod_{j = 1}^{k} \delta_{v_{j}} * \xi \right) \times \left( \prod_{j=k+1}^N \delta_{v_j} \right).
\]
Finally, let $\pi \in \mathbb{R}^{M}$ be a probability vector.
We call the measure 
\begin{equation}
 \label{eq-CantorMixture}
 \sum_{m = 1}^{M} \pi_{m}\xi_{v^{(m)}, N_m}
\end{equation}
a {\em Cantor mixture}.

Basically, for each $1 \le m \le M$, we are taking the $N_{m}$-fold product of the Cantor set with itself, embedding it into the first $N_{m}$ coordinates of $\mathbb{R}^{N}$, and translating it by the vector $v^{(m)}$.
The measure $\xi_{v^{(m)}, N_m}$ can be constructed on this set in the same way that we constructed the Cantor measure $\xi$.
Finally, in order to get the Cantor mixture, we sample from $\xi_{v^{(m)}, N_m}$ with probability $\pi_{m}$.
\\

In our analysis, we always choose 
\[
 v^{(m)} = (m, m, \hdots, m) \in \mathbb{R}^{N}.
\]
We generate five data sets:
\begin{enumerate}
 \item {\bf Data set A} consists of 10,000 points sampled from the Cantor measure up to 30 ternary digits of precision.
 \item {\bf Data set B} consists of 10,000 points sampled from the mixture of the Cantor measure and the five-fold product of the Cantor measure with weights $2/3$ and $1/3$ respectively.
 \item {\bf Data set C} consists of 2000 points sampled from the mixture of the two-fold and three-fold products of the Cantor measure with weights $3/4$ and $1/4$ respectively.
 \item {\bf Data set D} consists of 10,000 points sampled from the mixture of the Cantor measure with its five-fold and ten-fold products with weights $1/7$, $2/7$, and $4/7$ respectively.
 \item {\bf Data set E} consists of 10,000 points sampled from the mixture of the Cantor measure with its two-fold, three-fold, and four-fold products, each with equal weight.
\end{enumerate}
The results of the analysis are presented in Section \ref{sec-BasicAnalysisResults}, and are summarized in Table \ref{tab-CantorMixtureSummary}.

\subsubsection{Transformation Invariance \label{sec-BasicAnalysisTransformation}}

To test the degree of invariance of our estimator to locally bi-Lipschitz injections, we use a class of measures which generalize the measures $\zeta$ that we used in the analysis of the GP estimator.
Rather than defining the measures here, we explicitly describe the corresponding sampling procedures.
\\

We begin with some number $k$ of linear functions $f_{1}( x ), f_{2}( x ), \hdots, f_{k}( x )$ which satisfies the condition
\[
 f_{j}( 1 ) \le f_{j+1}( 0 ), 1 \le j < k.
\]

To generate a sample, we begin by first drawing a sample $X$ from the Cantor measure.
With probability $\pi_{k}$, our sample from the testing measure is $f_{k}( X )$.
\\

We begin with a probability vector $\pi = ( \pi_{1}, \pi_{2}, \hdots, \pi_{k} )$. We define, for $1 \le j \le k$, 
\[
 \sigma_{j} := \sum_{i=1}^{j} \pi_{i}
\]
and define $\sigma_{0} := 0$.
We consider the linear functions $f_{j}( x )$ of slope $2^{j-1}$ satisfying 
\[
 f_{j}( \sigma_{j} ) = f_{j+1}( \sigma_{j} ), 1 \le j < k.
\]
We draw a point $X$ from the Cantor measure.
With probability 1, there is unique $1 \le j \le k$ such that $X \in [ \sigma_{j-1}, \sigma_{j} ]$.
The corresponding sample from the testing measure, then, is $f_{j}( X )$.
\\

In our analysis, $k$ ranges from 1 to 5 and, for each $k$, we generate a hundred data sets, with each data set consisting of 10,000 samples.
In each case, we choose the probability vector $\pi$ by drawing at random from the space of $k$-dimensional probability vectors according to the Dirichlet distribution with the parameter $(1/2, 1/2, \hdots, 1/2 )$.
We choose linear functions $f_{j}( x )$ of slope $2^{j-1}$.
The results of our analysis are presented in the following section, and are summarized in Figure \ref{fig-TransformationInvariance}.

\subsection{Results \label{sec-BasicAnalysisResults}}

In each of the following analyses, we set the number of {\em Euclidean} nearest neighbors $n=1$, and the initial number of clusters $M=1$.
The parameters for the prior distributions ((\ref{eq-PriorDensity}), (\ref{eq-PriorDimension}), and (\ref{eq-PriorWeight})) are fixed to be $\alpha_{\delta} = \alpha_{d} = \gamma_{1}= 1$, and $\beta_{\delta} = \beta_{d} = 10^{-4}$ respectively. 
These are the only choices that we had to make for our estimator and, since our algorithm adaptively chooses the number of clusters $M$, our initialization $M=1$ is not really a serious limitation.
\\

We begin by presenting the results of our analysis on the Cantor mixture data sets which we described at the end of Section \ref{sec-BasicAnalysisDimensionBlindness}.
Table \ref{tab-CantorMixtureSummary} contains a summary which compares estimates for the average pointwise dimension of generating measures using our estimator on the given data sets with their true values.
We also show in the table the number of clusters chosen by our estimator for each data set, and the reader will see that this number of clusters agrees with the number of measures in each mixture.
In fact, the cluster assignments themselves are in line with the decompositions of each of the Cantor mixtures into combinations of the products of the Cantor measure.
Figure \ref{fig-ClusterResults} demonstrates that this is true in the cases of data sets D and E.
In the figure, the points in the data sets are listed along the vertical axis, and the clusters are marked along the horizontal axis.
The intensity of the color corresponding to a particular data point - cluster pair indicates the estimated probability that the given data point belong to the given cluster, with 0 being white and 1 being black.
Of course, our algorithm estimates the average pointwise dimension in each cluster separately. 
The results for each individual cluster in data sets D and E are presented in Table \ref{tab-CantorMixture}.
\\

The results involving the Cantor mixtures that we present exhibit the following general trends, which we have verified in more extensive experiments: 
\begin{enumerate}
 \item The errors in our estimates increase with the number of components  in the mixtures.
 \item The errors in our estimates for $n$-fold products of the Cantor measure increase with $n$.
 \item The errors in our estimates decrease as the sample size increases.
 \item The more components of similar dimension that there are in a mixture, the larger we expect the errors in our estimates to be.
\end{enumerate}
Of these four trends, the first three are not at all surprising.
The fourth calls for elaboration.
The clustering data for data set E provided in Table \ref{tab-CantorMixture} and Figure \ref{fig-ClusterResults} shows that the estimator mistakes points sampled from the two-fold and four-fold products of the Cantor measure as having been sampled from the three-fold product. 
Note that there is considerable physical separation in the data sampled from the various components of the mixture. 
Rather, this estimation error is a result of there being much more similarity between these components in terms of the nearest neighbor distances.
\\

The problem of reducing the interface between clusters as far as the nearest neighbor data is concerned does not seem to be out of reach.
There are many paths that one could follow to attack this problem -- the choice of metric used to calculate the nearest neighbor distances surely has a significant effect, for example, or one might be able to simply analyze the results of our estimator on each individual cluster proposed for the data set as a whole in order to get an idea of how reliable the original cluster assignments are.
Regardless of the method, any improvement along these lines seems to be highly dependent on the particular data in question.
We feel that this is an important consideration in data analysis, and we hope that our estimator provides a useful tool for its solution.
\\

\begin{table*}
 \begin{tabular}{ | c | c | c | c | c | c | c | c | }
 \hline
 \multicolumn{ 6 }{|c|}{ Data settings } & \multicolumn{ 2 }{ | c | }{ Estimation } \\
 \hline
 Dataset & Sample size & Clusters $M$ & $N_{1}, \ldots, N_{M}$ & Frequencies $\pi$ & Avg. Pointwise Dim. & Avg. Pointwise Dim. & Clusters $M$ \\
 \hline
 \hline
 A & 10,000 & 1 & 1 & $(1)$ & 0.63093 & 0.63063 & 1 \\
 \hline
  B & 10,000 & 2 & 1, 5 & $(2/3, 1/3)$ & 1.4722 & 1.4474 & 2 \\
 \hline
  C & 2,000 & 2 & 2, 3 & $(3/4, 1/4)$ & 1.4196 & 1.4449 & 2 \\
 \hline
  D & 10,000 & 3 & 1, 5, 10 & $(1/7, 2/7, 4/7)$ & 4.5968 & 4.6875 & 3 \\
 \hline
  E & 10,000 & 4 & 1, 2, 3, 4 & $(1/4, 1/4, 1/4, 1/4)$ & 1.5773 & 1.4460 & 4 \\
 \hline
 \end{tabular}
 \caption{\label{tab-CantorMixtureSummary} The summary of the analysis of the five Cantor mixture datasets}
\end{table*}


\begin{table}
 \begin{tabular}{ | c | c | c | c | }
 \hline
 Dataset D & \multicolumn{ 3 }{|c|}{ Cluster Index } \\
 \hline 
 &  1 & 2 & 3  \\
  \hline \hline 
 Estimated Weight &  0.1403  &  0.2798  &  0.5799\\
 \hline
 Estimated Dimension &  0.6293  &  3.2120  &  6.3811 \\
 \hline
 S.D. of Dimension & 0.0024  &  0.0090 &   0.0166 \\ 
 \hline \hline
 True Weight &  0.1429  &  0.2857  &  0.5714 \\
 \hline
 True Dimension &  0.6309  &  3.1546  &  6.3093 \\
 \hline
 \end{tabular} 
 
 \begin{tabular}{ | c | c | c | c | c | }
 \hline
  Dataset E & \multicolumn{ 4 }{|c|}{ Cluster Index } \\
 \hline 
 &  1 & 2 & 3 & 4 \\
  \hline \hline 
 Estimated Weight & 0.2572  &  0.2069 &   0.3299  &  0.2059 \\
 \hline
 Estimated Dimension &  0.6187  &  1.3876  &  1.4180  &  2.5825 \\
 \hline
 S.D. of Dimension & 0.0017  &  0.0037  &  0.0047  &  0.0084 \\ 
 \hline \hline
 True Weight &  0.25 &   0.25  &  0.25  &  0.25 \\
 \hline
 True Dimension &  0.6309 &   1.2619  &  1.8928  &  2.5237 \\
 \hline
 \end{tabular}
 \caption{\label{tab-CantorMixture} The details of the estimated parameters for Data set D and Data set E. The clusters are indexed in ascending order of its estimated average dimension within the cluster.}
\end{table}

 \begin{figure}[tb, clip]
   \begin{center}
    \includegraphics[width=1\linewidth, clip]{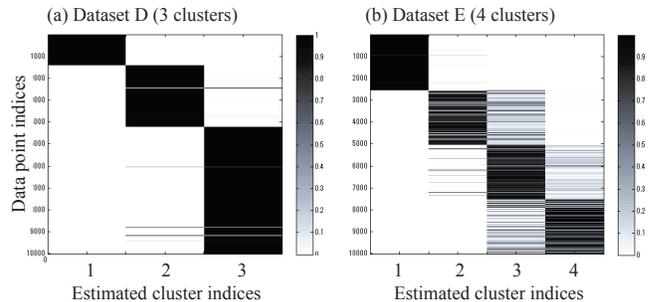}
    \caption{ The cluster assignment visualized by the gray-scale color map. The black cell indicates probability 1, while one indicate probability 0, and gray indicates in-between. For each case, the 10,000 data points are indexed in an ascending order of its true dimension. The clusters are indexed corresponding to Table \ref{tab-CantorMixture}.
    \label{fig-ClusterResults}}      
    \end{center}
  \end{figure}

In our analysis of the sensitivity of our estimator to locally bi-Lipschitz injections, we used the data sets described at the end of Section \ref{sec-BasicAnalysisTransformation}.
The results are summarized in Figure \ref{fig-TransformationInvariance}.
Each of the generating measures has uniform pointwise dimension $\log( 2 )/ \log( 3 ) \approx 0.6309$.
This is indicated in the figure by the red line.
The horizontal axis in the figure represents the number $k$ of regions of varying density (these are the images of the various functions $f_{j}$). The vertical axis represents the estimated average pointwise dimension. The vertical bars around our estimates reflect the standard error of the estimated average pointwise dimensions.
Unsurprisingly, the estimates get worse as the number of regions increases. 
However, even in the case $k=5$, the estimate obtained from our estimator is not far from the true value $\log( 2 )/ \log( 3 )$.
We plan to expand upon this issue of transformation invariance in a subsequent paper.
\\



 \begin{figure}[htb, clip]
   \begin{center}
    \includegraphics[width=1\linewidth, clip]{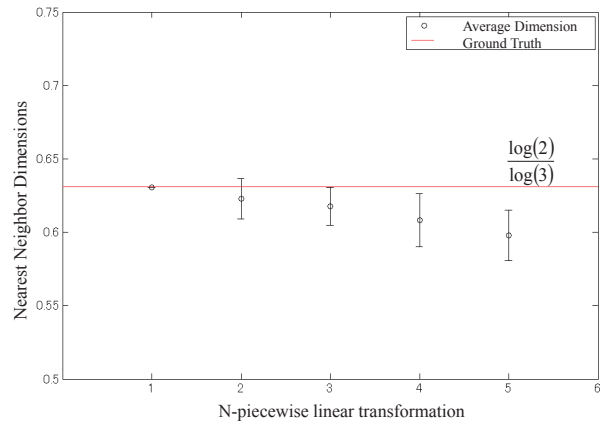}
    \caption{ Transformation invariance \label{fig-TransformationInvariance}}      
    \end{center}
  \end{figure}

With this, we conclude our discussion of the estimator itself.
In the next section we describe some potential uses of our estimator as well as some questions which we feel are important and which arise naturally from our discussion until this point.

\section{Open Questions \label{sec-Future}}

\subsection{Overview \label{sec-FutureOverview}}

We have presented in this paper a new technique for the numerical estimation of fractal dimensions.
As such, the majority of the potential for our method lies in applying it to numerical data.
However, in the development of our estimator, several fundamental questions about dimension estimation itself were raised, which share an intimate relationship with these potential applications.
These questions concern the applicability of our method and dimension estimation methods in general.
We will discuss these methods before moving on to a discussion of the potential numerical scope of our algorithm.
There are also several technical improvements that can potentially be made to the estimator that we have presented.
Suggestions for such improvements have been scattered throughout the text.
We omit such questions in this section, preferring to focus on what we consider to be the major challenges that our estimator makes accessible.
\\

This section is really more about what we do not know than what we know.
We do attempt, however, to provide preliminary numerical results with each category of open questions that we discuss.
Although we have presented our questions in different categories, there is considerable interaction between them.
This will be apparent from the number of cross references as far as our numerical analyses are concerned.

\subsection{Fundamental Challenges in Dimension Estimation \label{sec-FutureFundamentalChallenges}}

One of the main sources of motivation in the development of our estimator was dimension blindness of correlation dimension. 
We observed, however, that our GP estimator was not truly blind to the distributions of pointwise dimensions.
This suggests that such an estimator does not truly reflect the correlation dimension of the generating measure of a given data set. 
The first fundamental question concerns the relationship between pointwise dimension and correlation dimension:
\begin{question}
\label{q-PointwiseCorrelation}
 Given an estimator of pointwise dimension, is it possible to derive an estimator of correlation dimension?
\end{question}

As we observed in our discussion of the phenomenon of dimension blindness of correlation dimension, the correlation dimension of a finite mixture of measures of uniform pointwise dimension is the minimal pointwise dimension of these measures.
This suggests an estimator of correlation dimension built upon our proposed estimator of pointwise dimension which simply returns the minimum average dimension amongst the clusters in our approximation (\ref{eq-AlgorithmOutput}).
\\

Things are not so simple in general, for it is possible that the lowest-dimensional cluster in our estimate is simply an artifact of noise in our data. Even if this is not the case, it is impossible to tell whether the generating measure is indeed a {\em finite} mixture such as the ones we have discussed. 
It is not clear that our observation holds in the same manner for more general classes of measures.
We will begin to address the problem of noisy data in the following sections.
The other problem, however, seems to require a more serious mathematical analysis.
A generic measure, to the extent that such an object exists, would not exhibit the same regularity as a finite mixture of measures of exact pointwise dimension.
\\

It is not clear how to attack the problem in general, for it requires one to relate the distribution of pointwise dimensions to the correlation dimension. 
Our knowledge of this relationship is still in infancy -- as far as we are aware, the state of the art is the work of Young \cite{Young1982} and Pesin \cite{Pesin1993}.
\\

Additionally, the estimation method we proposed above for finite mixtures of measures of exact dimension suggests another interesting question -- does the minimal cluster dimension when using the $n^{\text{th}}$-nearest neighbor distances provide an estimate in fact for the $(n+1)$-generalized dimension of the mixture?
This relationship is suggested by (\ref{eq-GeneralizedDimension}) and (\ref{eq-PointwiseDimensionLimit}).
\\

The next big development in our derivation of the estimator was the limit-free description of pointwise dimension which we utilized.
We framed this description as holding for measures which satisfy the local uniformity condition (\ref{eq-LocalUniformity}). Thus, our algorithm as we present it only {\em guaranteed}  to return a meaningful estimate if the measure whose pointwise dimension distribution we are trying to estimate is exactly of the form (\ref{eq-AlgorithmOutput}).
It is, however, possible to show that the estimates produced by our algorithm are valid even for a more general class of measures -- we plan to discuss this in a future paper.
This discussion, however, leads to our second fundamental question:
\begin{question}
\label{q-Applicability}
 What is the most general class of measures for which one can give a limit-free description of pointwise dimension?
\end{question}

Such a limit-free description would not necessarily imply an effective estimator for the class in question.
It would, however, probably lead to a more generally applicable estimator than ours.
\\

There seems to be the measures for which such a description does not seem possible.
For example, those of Billingsley \cite{Billingsley1960}, which Cutler \cite{Cutler1993} terms ``fractal measures''.
For these measures, there is a dense set of points of measure 0 on which the pointwise dimension seems to vary wildly from the constant almost everywhere pointwise dimension of the measure. As far as our estimator is concerned, since we use the Poisson process formalism and the data can only be generated up to a finite scale, these points seems to be given undue weight. In fact, our estimator seems to  estimate the Hausdorff dimension of the support of such a measure.
This raises our third fundamental question, which is really more of a challenge:
\begin{question}
\label{q-Billingsley}
 Can one devise an estimator of pointwise dimension which approximates with reasonable accuracy the uniform pointwise dimensions of Billingsley's measures?
\end{question}

\subsection{Dimension and Dynamical Systems \label{sec-FutureDynamicalEstimation}}


The problems we discuss in this section are natural extensions of our discussion of the utility of Hausdorff dimension in the field of dynamical systems from Section \ref{sec-HausdorffDimension}.
We pose some of these questions here as questions about the numerical estimation of pointwise dimension distribution for the H\'enon attractors.
It is important to note, though, that these questions about H\'enon attractor have a much more general scope.
\\

Before we begin estimating, it is important to ask ourselves whether the data upon which we are running our estimator satisfies our hypotheses.
We have already discussed in the previous section problems involving generating measures which do not satisfy our uniformity condition.
In discussion dimension estimation for dynamical data, we are confronted with even more serious issue -- there need not even be a measure that one could consider as having generated the data in question. As far as we are aware, previous estimates of fractal dimension for attractors in dynamical systems have been obtained under the assumption that the generating system is ergodic.
In this case, the data can be assumed to have been generated by an ergodic measure for the given system.
There has, to our knowledge, been no systematic way of testing this hypothesis.
\\

The results of Cutler \cite{Cutler1990,Cutler1992}, of which we proved a simplified version as Proposition \ref{prop-ErgodicityExactDimension},
indicate the possibility of testing the ergodicity of certain classes of dynamical systems.
For example, as per Proposition \ref{prop-ErgodicityExactDimension}, a locally bi-Lipschitz injection $f$ cannot be ergodic if a reliable estimate of the pointwise dimension distribution of data from a generic forward orbit of $f$ indicates that the pointwise dimension is not constant.
The crucial point here is that the estimate must be {\em reliable}.
This shifts some of the difficulty of testing ergodicity to the more practical problem of establishing the reliability of an estimator.
\begin{question}
 \label{q-Ergodicity}
  Can one design an estimator of pointwise dimension which is reliable enough to falsify the hypothesis of ergodicity?
\end{question}

In the rest of this section, we use our estimator to analyze data generated from the H\'enon map with the goal of testing its ergodicity.
We choose the H\'enon map because it is a popular system to analyze in the field of dimension estimation.
It is not clear to us that our estimator is reliable enough for this purpose, and this is something we urge the reader to keep in mind during the analysis.
\\

For $a, b \in \mathbb{R}$, we define 
\begin{equation}
 \label{eq-HenonMap}
 F_{a,b}( x, y ) := ( y + 1 - a x^{2}, b x ).
\end{equation}
The maps $F_{ a, b }$ are called the {\em H\'enon maps} and were first studied by H\'enon \cite{Henon1976}.
The map $F := F_{ 1.4, 0.3 }$ is the {\em classical H\'enon map} (Figure \ref{fig-HenonClassical}).
 \begin{figure}[htb, clip]
   \begin{center}
    \includegraphics[width=1\linewidth, clip]{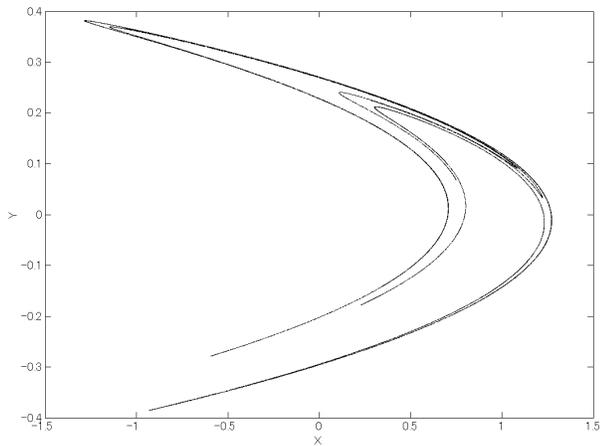}
    \caption{Attractor of the classical H\'enon map.
    \label{fig-HenonClassical} }      
    \end{center}
  \end{figure}
\\

The dynamics of this classical H\'enon map $F$ are of particular interest as the forward orbits this map converges to what is known as a {\em strange attractor}, which indicates that the dynamics are chaotic. 
We refer to the classical H\'enon map simply as the H\'enon map, and its attractor as the H\'enon attractor.
\\

The strangeness of the H\'enon attractor can be established, under the assumption of ergodicity, by estimating certain fractional dimensions associated with it.
For example, Russel, Hanson, and Ott \cite{Russell1980} estimated the box-counting dimension of the attractor to be $1.261 \pm 0.003$. Grassberger and Procaccia \cite{Grassberger1983} estimated the correlation dimension of this attractor to be $1.25 \pm 0.02$.
\\

When $b \neq 0$, the map $F_{ a, b }$ is clearly locally bi-Lipschitz -- its Jacobian at every point has determinant $b$. Thus, at least on these grounds, there can be no objection to using our estimator to test the ergodicity of the H\'enon map.
The entire weight of our conclusions rests upon the reliability of our estimator.
\\

To approximate the H\'enon attractor, we choose an initial point $( x_{0}, y_{0} )$ in the plane by sampling from a bivariate normal distribution and generate its forward orbit under the H\'enon map. As the attractor is only observable in the long-term behavior of the orbit, we discard the first 1000 iterates. We use the subsequent 30,000 iterates as our data.
In our analysis, we use three sets of data generated in this manner.
\\

First we present the estimate obtained by setting the number of Euclidean nearest neighbors $n$ to 1 and initializing our estimator with one cluster.
Our estimator found two clusters in each data set, and Figure \ref{fig-HenonNN} shows the pointwise dimension distribution for each cluster in Data set 4.
This yields an estimate of $1.0918$ for the average pointwise dimension of this data. 
The estimates for the individual clusters in the individual data sets seems to be much more revealing.
This information is presented in Table \ref{tab-HenonDimension}.
From these estimates, applying the principle of Proposition \ref{prop-ErgodicityExactDimension}, it seems highly unlikely that the H\'enon map is ergodic.
\\

  \begin{figure}[htb, clip]
   \begin{center}
    \includegraphics[width=1\linewidth, clip]{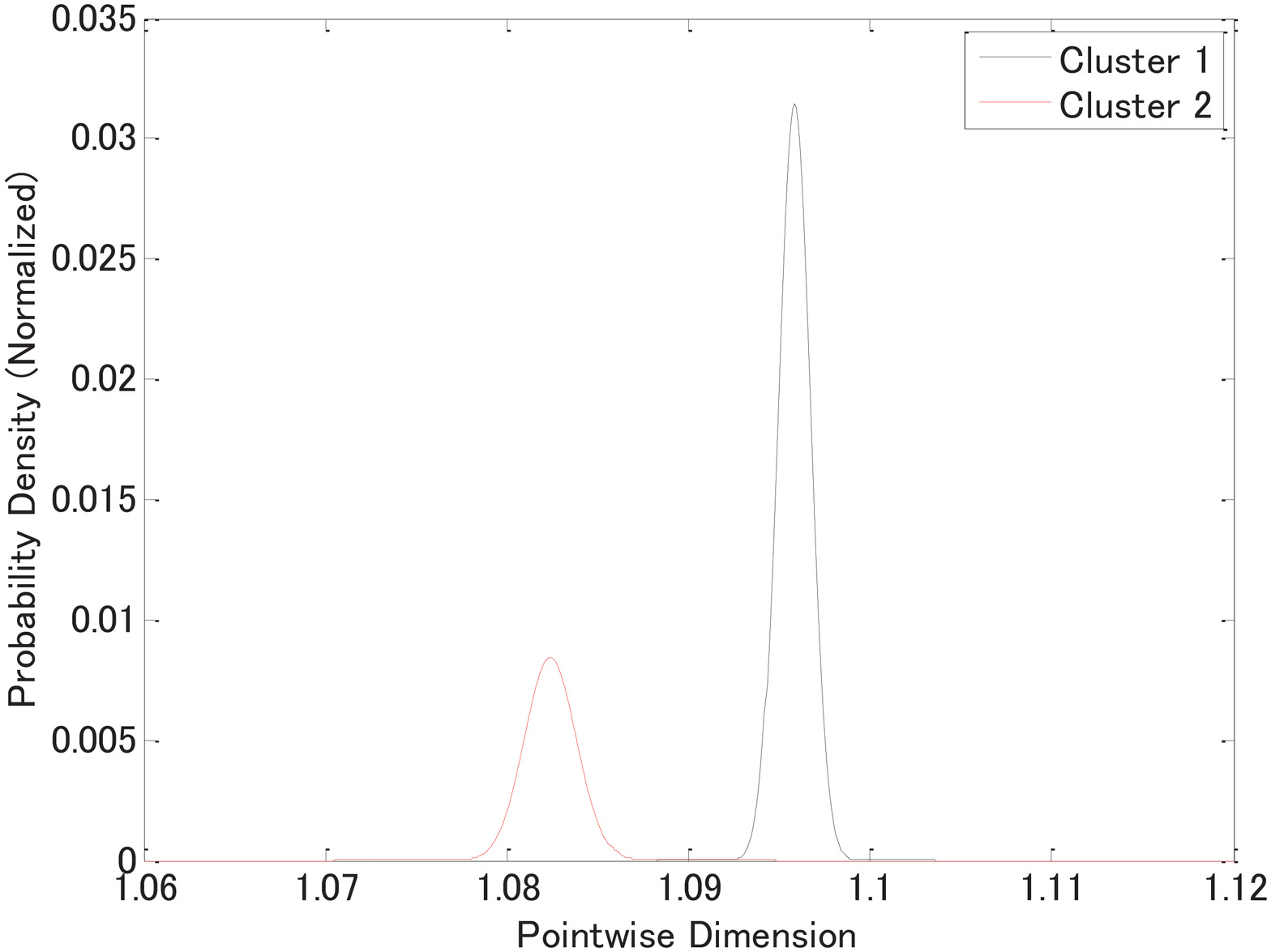}
    \includegraphics[width=1\linewidth, clip]{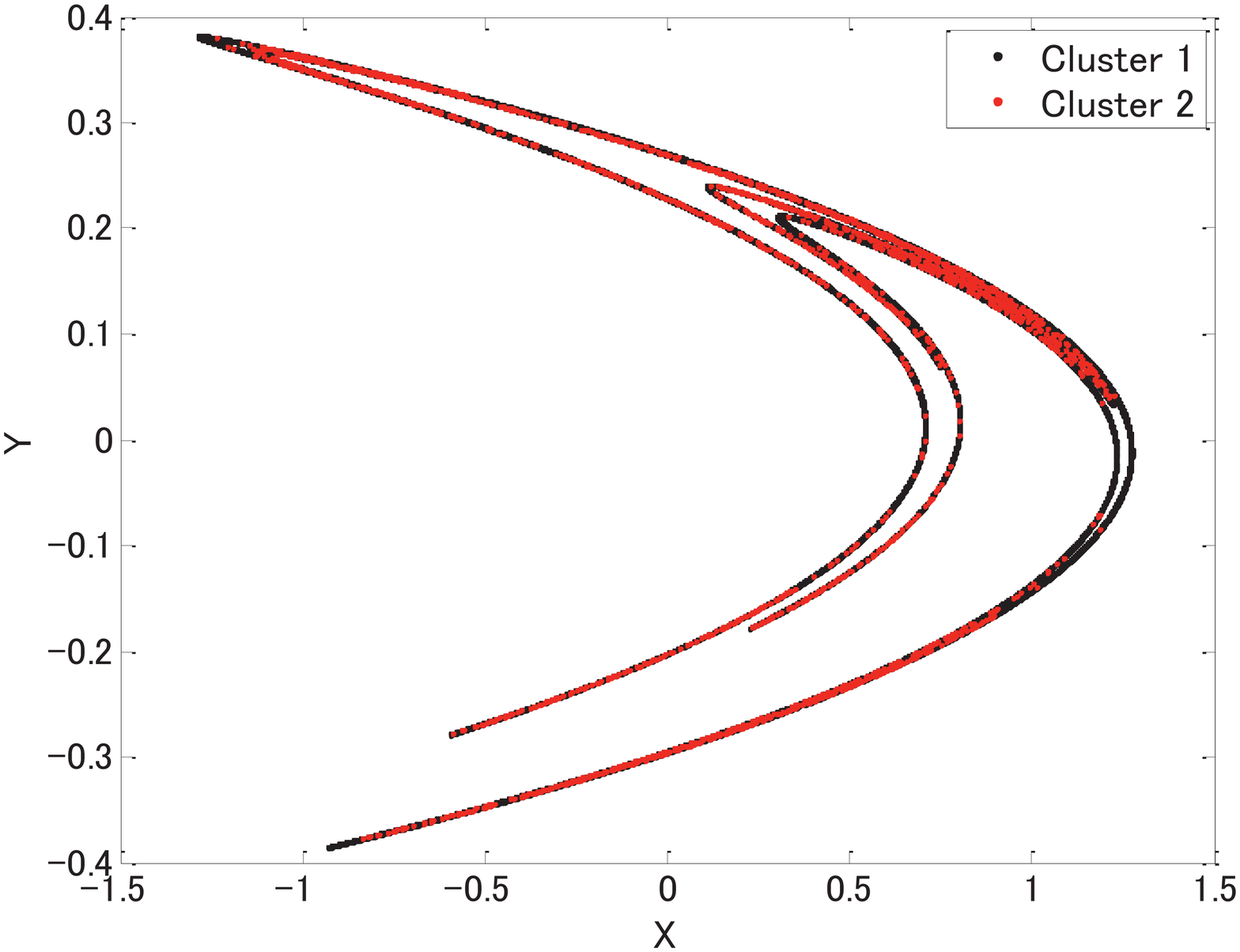}
    \caption{Top: Estimated pointwise dimension distribution of each cluster in a forward orbit of the H\'enon map. Bottom: The cluster assignments corresponding to the estimates indicated in the top figure.
    \label{fig-HenonNN} }      
    \end{center}
  \end{figure}

\begin{table}
 \begin{tabular}{| c | c | c | c | c |}
  \hline
            & {Data set 1} & {Data set 2} & {Data set 3} & {Data set 4} \\
  \hline \hline
  {Cluster 1} & 1.2327     &  1.1755    &    1.1571  & 1.0958 \\
  \hline
  {Cluster 2} & 1.1284     &  1.1137   &    1.1138 & 1.0824 \\
  \hline
 \end{tabular}
 \caption{\label{tab-HenonDimension} The average pointwise dimension of each cluster in each data set.}
\end{table}

It is worth noting that neither of the two estimated clusters is negligible in any of the data sets. This is indicated in Table \ref{tab-HenonWeight} which contains the proportions of data points in each data set belonging to each of the clusters.
It is also interesting that the average pointwise dimension of Cluster 2 seems to be fairly uniform across the data sets.
The reason for this is unclear to us.
\\

\begin{table}
 \begin{tabular}{| c | c | c | c | c |}
  \hline
            & {Data set 1} & {Data set 2} & {Data set 3} & {Data set 4} \\
  \hline \hline
  {Cluster 1} & 0.3236     &  0.1837   &  0.4191 & 0.6971\\   
  \hline
  {Cluster 2} & 0.6764    &  0.8163   &  0.5809 & 0.3029 \\   
  \hline
 \end{tabular}
  \caption{\label{tab-HenonWeight} The proportion of points in each data set belonging to each of the two clusters.}
\end{table}

It is possible that the variation between the data sets is the results of the randomness in our choice of initial points rather than reflection of chaotic behavior of the H\'enon map.
We are not aware of any serious analysis of the length of time it takes for iterates of the H\'enon map to converge to its attractor given an initial value, but this is certainly a pertinent question given the results of our analysis -- an estimate of this time to convergence given a particular initial value would make our estimator more reliable for the purpose of falsifying ergodicity of the H\'enon map.
One could perhaps conduct systematic numerical experiments with varying initial values, using our estimator as a tool to study this question.
\\

In the absence of any guarantee of a generating measure related to the H\'enon map for our data, it is very difficult to distinguish between essential dynamical properties of the system and artifacts related to these initial value issues.
We can, however, offer some evidence that the estimates produced by our estimator reflect the essential dynamics of the H\'enon map -- our estimates remain relatively stable as we take larger- and larger-dimensional time-delay embeddings of the data sets. 
This is indicated in Figure \ref{fig-HenonEmbedding}.
Such stability is commonly seen as evidence for low-dimensional dynamical behavior.
One would perhaps expect a more marked change in behavior as the embedding dimension was increased if the initial value effect is significant. 
The reasoning here is very speculative should not be taken seriously.
There are many potential contributing factors to the stability of our estimates under time delay embedding. For example, there could be a significant contributions from effects arising from the embedding maps themselves. 
We do not know how to rule these out. The matter calls for more expertise than we currently have.
\\

 \begin{figure}[htb, clip]
   \begin{center}
    \includegraphics[width=1\linewidth, clip]{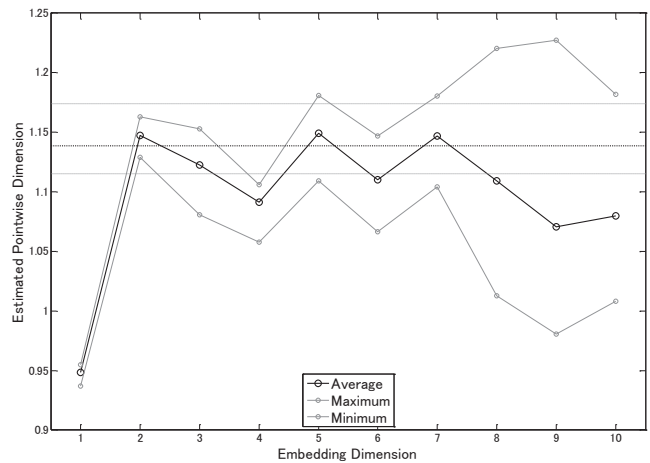}
    \caption{Estimated average pointwise dimension of the H\'enon attractor as function of embedding dimension. The broken lines indicate the estimated pointwise dimension of the H\'enon attractor on the original plane.
    \label{fig-HenonEmbedding} }      
    \end{center}
  \end{figure}
  
The stability that our estimates show under an increase in the embedding dimension is not present as we increase the number $n$ of nearest neighbors we use to produce our estimates.
This instability is indicated in Figure \ref{fig-HenonNNIncrease}.
We are not sure why this is so, but  it is perhaps related to the question of relationship between $n$ and the $(n+1)$-generalized dimension which arose in the discussion following Question \ref{q-PointwiseCorrelation}.
\\

   \begin{figure}[htb, clip]
   \begin{center}
    \includegraphics[width=1\linewidth, clip]{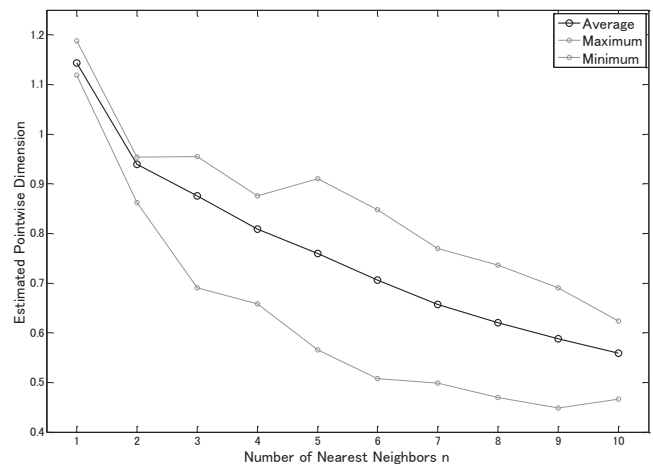}
    \caption{Estimated average pointwise dimension of the H\'enon attractor as a function of the number of nearest neighbors $n$.
    \label{fig-HenonNNIncrease} }      
    \end{center}
  \end{figure}

Finally, although it is not known whether or not the H\'enon map is ergodic, 
Benedicks and Young \cite{Young1993} have proved that there is a set of parameters $(a, b)$ of positive Lebesgue measure for which there  exists a neighborhood of the attractor for the map $F_{a, b}$ in which, for almost every point $x$ with respect to the Lebesgue measure, the time averages of continuous functions $\phi$ on the neighborhood converge to 
\[
 \frac{ 1 }{ n } \sum_{ j = 1 }^{ n } \phi( F_{ a, b }^{ j }( x ) ) \rightarrow \int \phi \D \lambda^{*},
\]
for some measure $\lambda^{*}$.
This measure $\lambda^{*}$ is called a Sinai-Bowen-Ruelle (SBR) measure for $F_{a,b}$.
\\

We are not aware of the current state of knowledge about this SBR measure $\lambda^{*}$.
If the contribution of initial value effects we have discussed above is relatively small, however, it is likely that the pointwise dimension distributions for our data sets reflect the pointwise dimension distribution of $\lambda^{*}$.
If this is the case, our estimator could prove to be a valuable tool in the study of not only $\lambda^{*}$ but SBR measures for general dynamical systems.
Before such an application becomes feasible, it is important to study the following question:
\begin{question}
 \label{q-SBR}
  What properties do the pointwise dimension distributions of Sinai-Bowen-Ruelle measures exhibit?
\end{question}

\subsection{Detection of Stochastic Effects \label{sec-FutureStochastic}}

There are many considerations which motivate the discussion in this section but, broadly speaking, this section is about noise.
We will use as illustrations two classes of data sets.

The first class consists of data sampled from the Cantor measure but with some added Gaussian noise.
Let $X$ be a random variable distributed according to the Cantor measure and let $Y_{\sigma}$ be a normal random variable with mean $0$ and variance $\sigma^{2}$. 
We construct the data set $\mathcal{C}_{\sigma}$ by sampling independently 10,000 points from $X + Y_{\sigma}$. Note that this is equivalent to sampling from the convolution of the Cantor measure with the corresponding normal measure. 

The second class consists of a single data set $\mathcal{B}$, which is generated by a Wiener process $W(t)$. $\mathcal{B}$ consists of the data $W(1), W(2), \hdots, W(K)$.
\\


We fix a positive integer $n$, the number of nearest neighbors we will use in our estimator.
We denote by $R$ the expectation of square of the distance from a point sampled from the Cantor measure and its $n^{\text{th}}$-nearest neighbor given that we are sampling a total of 10,000 such points.
For given value of $\sigma$, we write 
\[
 \lambda_{\sigma} := \frac{ 2 \sigma^{2} }{ R }.
\] 
We call $\lambda_{\sigma}$ the {\em noise level} corresponding to $\sigma$.
It is the expected (additive) surplus when we divide the square of the expected $n^{\text{th}}$-nearest neighbor distance of the noisy data with that of the noise-free data.
\\

In our analysis, we used data sets $\mathcal{C}_{\sigma}$ with noise levels 
\[
 \lambda_{\sigma} = 10^{-6 + \frac{k}{4}} ,
\]
for $0 \le k \le 56$.
In Figure \ref{fig-NoiseAnalysisDetail}, we present the results of this analysis.
We fixed the parameter $n=1$.
The open figures indicate the number of clusters detected at the corresponding noise level by their shape, as indicated in the legend.
The filled figures correspond to us forcing upon our estimator a fixed number of clusters.

 \begin{figure}[htb, clip]
   \begin{center}
    \includegraphics[width=1\linewidth, clip]{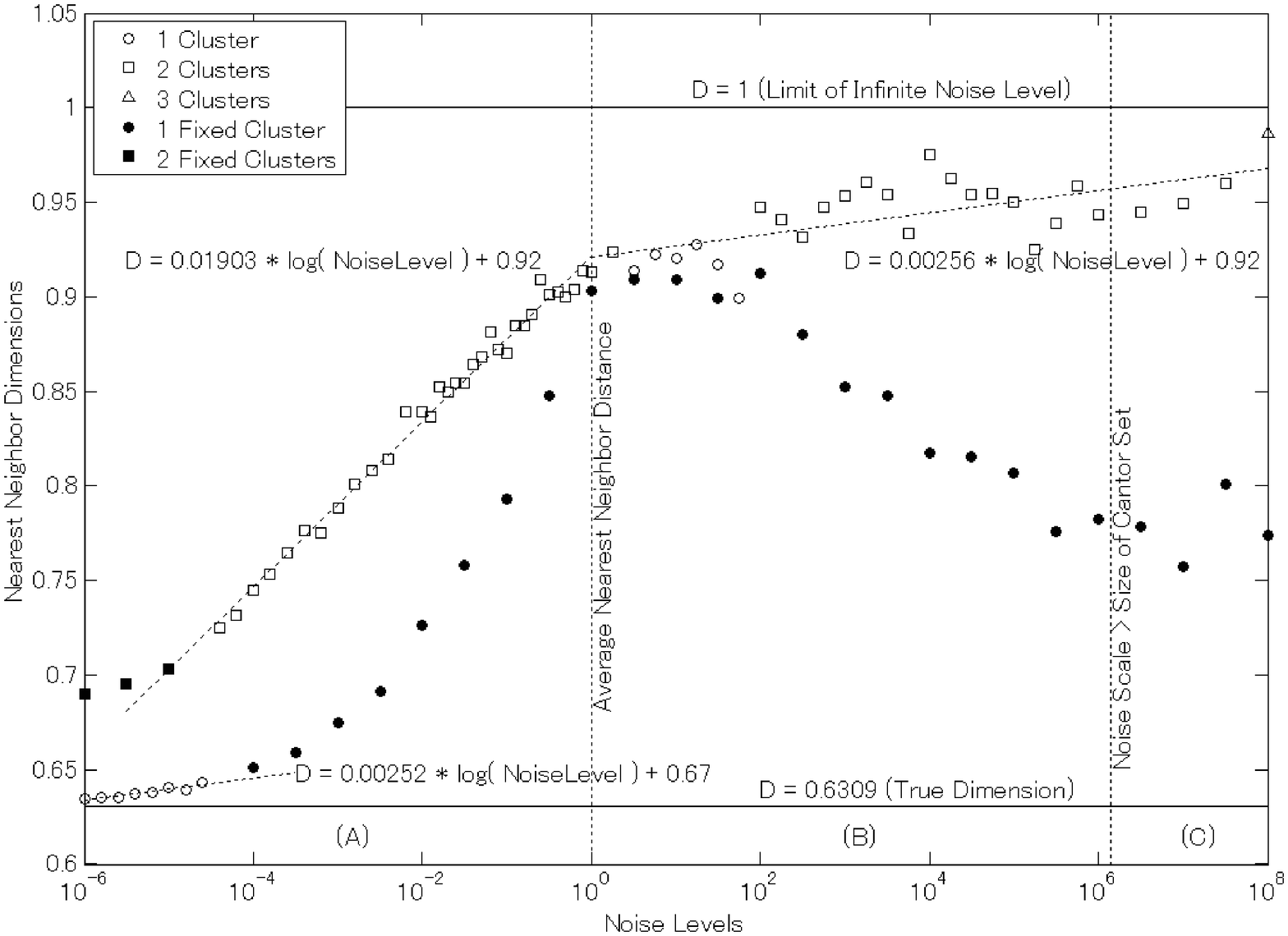}
    \caption{Estimated average pointwise dimension of noisy Cantor measures.
    \label{fig-NoiseAnalysisDetail} }      
    \end{center}
  \end{figure}

There are three distinct intervals of the noise levels: 
\begin{itemize}
\item[(A)] The interval where the noise level is smaller than the average nearest neighbor distance. Here, the structure of the original Cantor set remains substantially unchanged.
\item[(B)] The interval where the noise level is higher than the average nearest neighbor distances but still overlaps with the support of the Cantor measure. Here, the noise begins to dominate the Cantor measure in terms of dimension.
\item[(C)] The noise level is larger than the length of the support of the Cantor measure and the structure of the Cantor measure is almost invisible in the dimension estimates.
\end{itemize}
Note that the local densities of the data points are not uniform at each of the noise levels in intervals (B) and (C). This is due to the relatively large scale of the noise at those levels as compared to the scale of the data generated from the Cantor measure.
It is interesting to note, however, that our estimator detected only one cluster at noise levels between 1 and 100.
\\

With the Brownian motion data $\mathcal{B}$, we analyzed the effects of increasing the time-delay embedding dimension.
The results of our analysis are presented in Figure \ref{fig-FractionalBrownianMotion}.
The independence properties of Browian motion dictate the growth (\cite{Cutler1993}; Theorem 4.3.2) of pointwise dimension observed in our estimates.

 \begin{figure}[htb, clip]
   \begin{center}
    \includegraphics[width=1\linewidth, clip]{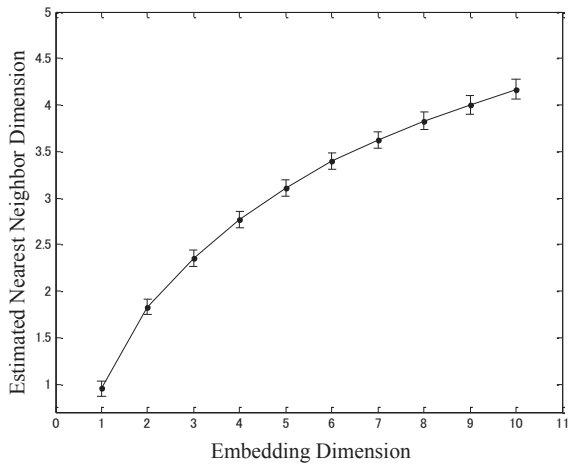}
    \caption{The average nearest neighbor dimensions estimated on a set of time series of fractional Brownian motion as a function of embedding dimension.
    \label{fig-FractionalBrownianMotion} }      
    \end{center}
  \end{figure}
  
This growth does not manifest itself in similar analyses with GP estimators  except under very limited circumstances.
This problem with GP estimators was discovered numerically by Osborne \cite{Osborne1989} and Rapp \cite{Rapp1993}.
Theiler \cite{Theiler1991} contains a theoretical discussion of this phenomenon.
Schreiber \cite{Schreiber1993} has analyzed the effects of normally distributed additive noise when data is embedded in unnecessarily high dimensions.
\\

As a word of warning, Cutler \cite{Cutler1993} in Remark 4.3.4 provides an example of a stochastic process which does not exhibit the same type of growth of dimension as we observed in the case of Brownian motion. The notion that one can distinguish between dynamic and stochastic behavior from such an analysis using embedding dimensions is only a rule of thumb.
\\

Still, our analysis of the noisy Cantor set data as well as the Brownian motion data suggests the following question:

\begin{question}
 \label{q-NoiseDetection}
  Can one devise a systematic method for noise detection given an effective estimator of pointwise dimension?
\end{question}

Furthermore, it may also be possible to use estimates of pointwise dimension to reduce noise in data.
Most existing methods of noise reduction, at least as far as time series data is concerned, involve the calculation of local coordinates for the time series at each data point.
See, for example, the papers of Sauer \cite{Sauer1992}, Kantz \cite{Kantz1993}, and Grassberger \cite{Grassberger1993}.
There seems to be some connection between pointwise dimension and the number of such parameters required at a given point.
This is essentially our final question:

\begin{question}
 \label{q-NoiseReduction}
 What is the relationship between pointwise dimension of measure at a given point and the number of parameters required to express the data near that point?
\end{question}

\bibliography{NNDimReferences}

\begin{thebibliography}{34}%
\makeatletter
\providecommand \@ifxundefined [1]{%
 \@ifx{#1\undefined}
}%
\providecommand \@ifnum [1]{%
 \ifnum #1\expandafter \@firstoftwo
 \else \expandafter \@secondoftwo
 \fi
}%
\providecommand \@ifx [1]{%
 \ifx #1\expandafter \@firstoftwo
 \else \expandafter \@secondoftwo
 \fi
}%
\providecommand \natexlab [1]{#1}%
\providecommand \enquote  [1]{``#1''}%
\providecommand \bibnamefont  [1]{#1}%
\providecommand \bibfnamefont [1]{#1}%
\providecommand \citenamefont [1]{#1}%
\providecommand \href@noop [0]{\@secondoftwo}%
\providecommand \href [0]{\begingroup \@sanitize@url \@href}%
\providecommand \@href[1]{\@@startlink{#1}\@@href}%
\providecommand \@@href[1]{\endgroup#1\@@endlink}%
\providecommand \@sanitize@url [0]{\catcode `\\12\catcode `\$12\catcode
  `\&12\catcode `\#12\catcode `\^12\catcode `\_12\catcode `\%12\relax}%
\providecommand \@@startlink[1]{}%
\providecommand \@@endlink[0]{}%
\providecommand \url  [0]{\begingroup\@sanitize@url \@url }%
\providecommand \@url [1]{\endgroup\@href {#1}{\urlprefix }}%
\providecommand \urlprefix  [0]{URL }%
\providecommand \Eprint [0]{\href }%
\providecommand \doibase [0]{http://dx.doi.org/}%
\providecommand \selectlanguage [0]{\@gobble}%
\providecommand \bibinfo  [0]{\@secondoftwo}%
\providecommand \bibfield  [0]{\@secondoftwo}%
\providecommand \translation [1]{[#1]}%
\providecommand \BibitemOpen [0]{}%
\providecommand \bibitemStop [0]{}%
\providecommand \bibitemNoStop [0]{.\EOS\space}%
\providecommand \EOS [0]{\spacefactor3000\relax}%
\providecommand \BibitemShut  [1]{\csname bibitem#1\endcsname}%
\let\auto@bib@innerbib\@empty
\bibitem [{\citenamefont {Grassberger}\ and\ \citenamefont
  {Procaccia}(1983)}]{Grassberger1983}%
  \BibitemOpen
  \bibfield  {author} {\bibinfo {author} {\bibfnamefont {P.}~\bibnamefont
  {Grassberger}}\ and\ \bibinfo {author} {\bibfnamefont {I.}~\bibnamefont
  {Procaccia}},\ }\href@noop {} {\bibfield  {journal} {\bibinfo  {journal}
  {Physica D: Nonlinear Phenomena}\ }\textbf {\bibinfo {volume} {9}},\ \bibinfo
  {pages} {189} (\bibinfo {year} {1983})}\BibitemShut {NoStop}%
\bibitem [{\citenamefont {Billingsley}(1960)}]{Billingsley1960}%
  \BibitemOpen
  \bibfield  {author} {\bibinfo {author} {\bibfnamefont {P.}~\bibnamefont
  {Billingsley}},\ }\href@noop {} {\bibfield  {journal} {\bibinfo  {journal}
  {Illinois Journal of Mathematics}\ }\textbf {\bibinfo {volume} {4}},\
  \bibinfo {pages} {187} (\bibinfo {year} {1960})}\BibitemShut {NoStop}%
\bibitem [{\citenamefont {Billingsley}(1961)}]{Billingsley1961}%
  \BibitemOpen
  \bibfield  {author} {\bibinfo {author} {\bibfnamefont {P.}~\bibnamefont
  {Billingsley}},\ }\href@noop {} {\bibfield  {journal} {\bibinfo  {journal}
  {Illinois Journal of Mathematics}\ }\textbf {\bibinfo {volume} {5}},\
  \bibinfo {pages} {291} (\bibinfo {year} {1961})}\BibitemShut {NoStop}%
\bibitem [{\citenamefont {Ledrappier}\ and\ \citenamefont
  {Misiurewicz}(1985)}]{Ledrappier1985}%
  \BibitemOpen
  \bibfield  {author} {\bibinfo {author} {\bibfnamefont {F.}~\bibnamefont
  {Ledrappier}}\ and\ \bibinfo {author} {\bibfnamefont {M.}~\bibnamefont
  {Misiurewicz}},\ }\href@noop {} {\emph {\bibinfo {title} {Dimension of
  invariant measures for maps with exponent zero}}}\ (\bibinfo  {publisher}
  {Cambridge Univ Press},\ \bibinfo {year} {1985})\BibitemShut {NoStop}%
\bibitem [{\citenamefont {Cutler}\ and\ \citenamefont
  {Dawson}(1990{\natexlab{a}})}]{Cutler1990}%
  \BibitemOpen
  \bibfield  {author} {\bibinfo {author} {\bibfnamefont {C.~D.}\ \bibnamefont
  {Cutler}}\ and\ \bibinfo {author} {\bibfnamefont {D.~A.}\ \bibnamefont
  {Dawson}},\ }\href@noop {} {\bibfield  {journal} {\bibinfo  {journal} {The
  Annals of Probability}\ ,\ \bibinfo {pages} {256}} (\bibinfo {year}
  {1990}{\natexlab{a}})}\BibitemShut {NoStop}%
\bibitem [{\citenamefont {Hausdorff}(1918)}]{Hausdorff1918}%
  \BibitemOpen
  \bibfield  {author} {\bibinfo {author} {\bibfnamefont {F.}~\bibnamefont
  {Hausdorff}},\ }\href@noop {} {\bibfield  {journal} {\bibinfo  {journal}
  {Mathematische Annalen}\ }\textbf {\bibinfo {volume} {79}},\ \bibinfo {pages}
  {157} (\bibinfo {year} {1918})}\BibitemShut {NoStop}%
\bibitem [{\citenamefont {Ott}(2002)}]{Ott2002}%
  \BibitemOpen
  \bibfield  {author} {\bibinfo {author} {\bibfnamefont {E.}~\bibnamefont
  {Ott}},\ }\href@noop {} {\emph {\bibinfo {title} {Chaos in dynamical
  systems}}}\ (\bibinfo  {publisher} {Cambridge university press},\ \bibinfo
  {year} {2002})\BibitemShut {NoStop}%
\bibitem [{\citenamefont {Pesin}(1993)}]{Pesin1993}%
  \BibitemOpen
  \bibfield  {author} {\bibinfo {author} {\bibfnamefont {Y.~B.}\ \bibnamefont
  {Pesin}},\ }\href@noop {} {\bibfield  {journal} {\bibinfo  {journal} {Journal
  of statistical physics}\ }\textbf {\bibinfo {volume} {71}},\ \bibinfo {pages}
  {529} (\bibinfo {year} {1993})}\BibitemShut {NoStop}%
\bibitem [{\citenamefont {Hentschel}\ and\ \citenamefont
  {Procaccia}(1983)}]{Hentschel1983}%
  \BibitemOpen
  \bibfield  {author} {\bibinfo {author} {\bibfnamefont {H.}~\bibnamefont
  {Hentschel}}\ and\ \bibinfo {author} {\bibfnamefont {I.}~\bibnamefont
  {Procaccia}},\ }\href@noop {} {\bibfield  {journal} {\bibinfo  {journal}
  {Physica D: Nonlinear Phenomena}\ }\textbf {\bibinfo {volume} {8}},\ \bibinfo
  {pages} {435} (\bibinfo {year} {1983})}\BibitemShut {NoStop}%
\bibitem [{\citenamefont {Cutler}(1993)}]{Cutler1993}%
  \BibitemOpen
  \bibfield  {author} {\bibinfo {author} {\bibfnamefont {C.~D.}\ \bibnamefont
  {Cutler}},\ }\enquote {\bibinfo {title} {A review of the theory and
  estimation of fractal dimension},}\ in\ \href@noop {} {\emph {\bibinfo
  {booktitle} {Dimension estimation and models}}},\ Vol.~\bibinfo {volume} {1}\
  (\bibinfo  {publisher} {World Scientific},\ \bibinfo {year} {1993})\ pp.\
  \bibinfo {pages} {1--107}\BibitemShut {NoStop}%
\bibitem [{\citenamefont {Young}(1982)}]{Young1982}%
  \BibitemOpen
  \bibfield  {author} {\bibinfo {author} {\bibfnamefont {L.-S.}\ \bibnamefont
  {Young}},\ }\href {\doibase 10.1017/S0143385700009615} {\bibfield  {journal}
  {\bibinfo  {journal} {Ergodic Theory and Dynamical Systems}\ }\textbf
  {\bibinfo {volume} {2}},\ \bibinfo {pages} {109} (\bibinfo {year}
  {1982})}\BibitemShut {NoStop}%
\bibitem [{\citenamefont {Cutler}(1992{\natexlab{a}})}]{Cutler1992}%
  \BibitemOpen
  \bibfield  {author} {\bibinfo {author} {\bibfnamefont {C.~D.}\ \bibnamefont
  {Cutler}},\ }in\ \href@noop {} {\emph {\bibinfo {booktitle} {Proceedings of
  the 1990 Measure Theory Conference at Oberwolfach. Supplemento Ai Rendiconti
  del Circolo Mathematico di Palermo, Ser. II}}},\ \bibinfo {series and number}
  {\bibinfo {number} {28}}\ (\bibinfo {year} {1992})\ pp.\ \bibinfo {pages}
  {319--340}\BibitemShut {NoStop}%
\bibitem [{\citenamefont {Cutler}\ and\ \citenamefont
  {Dawson}(1989)}]{Cutler1989}%
  \BibitemOpen
  \bibfield  {author} {\bibinfo {author} {\bibfnamefont {C.~D.}\ \bibnamefont
  {Cutler}}\ and\ \bibinfo {author} {\bibfnamefont {D.~A.}\ \bibnamefont
  {Dawson}},\ }\href@noop {} {\bibfield  {journal} {\bibinfo  {journal}
  {Journal of multivariate analysis}\ }\textbf {\bibinfo {volume} {28}},\
  \bibinfo {pages} {115} (\bibinfo {year} {1989})}\BibitemShut {NoStop}%
\bibitem [{\citenamefont {Cutler}\ and\ \citenamefont
  {Dawson}(1990{\natexlab{b}})}]{Cutler1990NN}%
  \BibitemOpen
  \bibfield  {author} {\bibinfo {author} {\bibfnamefont {C.~D.}\ \bibnamefont
  {Cutler}}\ and\ \bibinfo {author} {\bibfnamefont {D.~A.}\ \bibnamefont
  {Dawson}},\ }\href@noop {} {\bibfield  {journal} {\bibinfo  {journal} {The
  Annals of Probability}\ ,\ \bibinfo {pages} {256}} (\bibinfo {year}
  {1990}{\natexlab{b}})}\BibitemShut {NoStop}%
\bibitem [{\citenamefont {Cutler}(1992{\natexlab{b}})}]{Cutler1992NN}%
  \BibitemOpen
  \bibfield  {author} {\bibinfo {author} {\bibfnamefont {C.~D.}\ \bibnamefont
  {Cutler}},\ }\enquote {\bibinfo {title} {$k$th nearest neighbor and the
  generalized logistic distribution},}\ \ (\bibinfo  {publisher} {Marcel
  Dekker},\ \bibinfo {year} {1992})\BibitemShut {NoStop}%
\bibitem [{\citenamefont {Badii}\ and\ \citenamefont
  {Politi}(1985)}]{Badii1985}%
  \BibitemOpen
  \bibfield  {author} {\bibinfo {author} {\bibfnamefont {R.}~\bibnamefont
  {Badii}}\ and\ \bibinfo {author} {\bibfnamefont {A.}~\bibnamefont {Politi}},\
  }\href@noop {} {\bibfield  {journal} {\bibinfo  {journal} {Journal of
  Statistical Physics}\ }\textbf {\bibinfo {volume} {40}},\ \bibinfo {pages}
  {725} (\bibinfo {year} {1985})}\BibitemShut {NoStop}%
\bibitem [{\citenamefont {Nerenberg}\ and\ \citenamefont
  {Essex}(1990)}]{Nerenberg1990}%
  \BibitemOpen
  \bibfield  {author} {\bibinfo {author} {\bibfnamefont {M.~A.~H.}\
  \bibnamefont {Nerenberg}}\ and\ \bibinfo {author} {\bibfnamefont
  {C.}~\bibnamefont {Essex}},\ }\href {\doibase 10.1103/PhysRevA.42.7065}
  {\bibfield  {journal} {\bibinfo  {journal} {Phys. Rev. A}\ }\textbf {\bibinfo
  {volume} {42}},\ \bibinfo {pages} {7065} (\bibinfo {year}
  {1990})}\BibitemShut {NoStop}%
\bibitem [{\citenamefont {Jordan}\ \emph {et~al.}(1997)\citenamefont {Jordan},
  \citenamefont {Ghahramani}, \citenamefont {Jaakkola},\ and\ \citenamefont
  {Saul}}]{Jordan1997}%
  \BibitemOpen
  \bibfield  {author} {\bibinfo {author} {\bibfnamefont {M.~I.}\ \bibnamefont
  {Jordan}}, \bibinfo {author} {\bibfnamefont {Z.}~\bibnamefont {Ghahramani}},
  \bibinfo {author} {\bibfnamefont {T.~S.}\ \bibnamefont {Jaakkola}}, \ and\
  \bibinfo {author} {\bibfnamefont {L.~K.}\ \bibnamefont {Saul}},\ }\href@noop
  {} {\emph {\bibinfo {title} {An introduction to variational methods for
  graphical models}}}\ (\bibinfo  {publisher} {Springer},\ \bibinfo {year}
  {1997})\BibitemShut {NoStop}%
\bibitem [{\citenamefont {Attias}(1999)}]{Attias1999}%
  \BibitemOpen
  \bibfield  {author} {\bibinfo {author} {\bibfnamefont {H.}~\bibnamefont
  {Attias}},\ }\href@noop {} {\bibfield  {journal} {\bibinfo  {journal}
  {Advances in neural information processing systems}\ }\textbf {\bibinfo
  {volume} {12}},\ \bibinfo {pages} {209} (\bibinfo {year} {1999})}\BibitemShut
  {NoStop}%
\bibitem [{\citenamefont {Beal}(2003)}]{Beal2003}%
  \BibitemOpen
  \bibfield  {author} {\bibinfo {author} {\bibfnamefont {M.~J.}\ \bibnamefont
  {Beal}},\ }\emph {\bibinfo {title} {Variational algorithms for approximate
  Bayesian inference}},\ \href@noop {} {Ph.D. thesis},\ \bibinfo  {school}
  {University of London} (\bibinfo {year} {2003})\BibitemShut {NoStop}%
\bibitem [{\citenamefont {Ghahramani}\ and\ \citenamefont
  {Beal}(1999)}]{Ghahramani1999}%
  \BibitemOpen
  \bibfield  {author} {\bibinfo {author} {\bibfnamefont {Z.}~\bibnamefont
  {Ghahramani}}\ and\ \bibinfo {author} {\bibfnamefont {M.~J.}\ \bibnamefont
  {Beal}},\ }in\ \href@noop {} {\emph {\bibinfo {booktitle} {NIPS}}}\ (\bibinfo
  {year} {1999})\ pp.\ \bibinfo {pages} {449--455}\BibitemShut {NoStop}%
\bibitem [{\citenamefont {Ueda}\ \emph {et~al.}(2000)\citenamefont {Ueda},
  \citenamefont {Nakano}, \citenamefont {Ghahramani},\ and\ \citenamefont
  {Hinton}}]{Ueda2000}%
  \BibitemOpen
  \bibfield  {author} {\bibinfo {author} {\bibfnamefont {N.}~\bibnamefont
  {Ueda}}, \bibinfo {author} {\bibfnamefont {R.}~\bibnamefont {Nakano}},
  \bibinfo {author} {\bibfnamefont {Z.}~\bibnamefont {Ghahramani}}, \ and\
  \bibinfo {author} {\bibfnamefont {G.~E.}\ \bibnamefont {Hinton}},\
  }\href@noop {} {\bibfield  {journal} {\bibinfo  {journal} {Neural
  computation}\ }\textbf {\bibinfo {volume} {12}},\ \bibinfo {pages} {2109}
  (\bibinfo {year} {2000})}\BibitemShut {NoStop}%
\bibitem [{\citenamefont {Ueda}\ and\ \citenamefont
  {Ghahramani}(2002)}]{Ueda2002}%
  \BibitemOpen
  \bibfield  {author} {\bibinfo {author} {\bibfnamefont {N.}~\bibnamefont
  {Ueda}}\ and\ \bibinfo {author} {\bibfnamefont {Z.}~\bibnamefont
  {Ghahramani}},\ }\href@noop {} {\bibfield  {journal} {\bibinfo  {journal}
  {Neural Networks}\ }\textbf {\bibinfo {volume} {15}},\ \bibinfo {pages}
  {1223} (\bibinfo {year} {2002})}\BibitemShut {NoStop}%
\bibitem [{\citenamefont {Dempster}\ \emph {et~al.}(1977)\citenamefont
  {Dempster}, \citenamefont {Laird},\ and\ \citenamefont
  {Rubin}}]{Dempster1977}%
  \BibitemOpen
  \bibfield  {author} {\bibinfo {author} {\bibfnamefont {A.~P.}\ \bibnamefont
  {Dempster}}, \bibinfo {author} {\bibfnamefont {N.~M.}\ \bibnamefont {Laird}},
  \ and\ \bibinfo {author} {\bibfnamefont {D.~B.}\ \bibnamefont {Rubin}},\
  }\href@noop {} {\bibfield  {journal} {\bibinfo  {journal} {Journal of the
  Royal Statistical Society. Series B (Methodological)}\ ,\ \bibinfo {pages}
  {1}} (\bibinfo {year} {1977})}\BibitemShut {NoStop}%
\bibitem [{\citenamefont {H{\'e}non}(1976)}]{Henon1976}%
  \BibitemOpen
  \bibfield  {author} {\bibinfo {author} {\bibfnamefont {M.}~\bibnamefont
  {H{\'e}non}},\ }\href@noop {} {\bibfield  {journal} {\bibinfo  {journal}
  {Communications in Mathematical Physics}\ }\textbf {\bibinfo {volume} {50}},\
  \bibinfo {pages} {69} (\bibinfo {year} {1976})}\BibitemShut {NoStop}%
\bibitem [{\citenamefont {Russell}\ \emph {et~al.}(1980)\citenamefont
  {Russell}, \citenamefont {Hanson},\ and\ \citenamefont {Ott}}]{Russell1980}%
  \BibitemOpen
  \bibfield  {author} {\bibinfo {author} {\bibfnamefont {D.~A.}\ \bibnamefont
  {Russell}}, \bibinfo {author} {\bibfnamefont {J.~D.}\ \bibnamefont {Hanson}},
  \ and\ \bibinfo {author} {\bibfnamefont {E.}~\bibnamefont {Ott}},\ }\href
  {\doibase 10.1103/PhysRevLett.45.1175} {\bibfield  {journal} {\bibinfo
  {journal} {Phys. Rev. Lett.}\ }\textbf {\bibinfo {volume} {45}},\ \bibinfo
  {pages} {1175} (\bibinfo {year} {1980})}\BibitemShut {NoStop}%
\bibitem [{\citenamefont {Benedicks}\ and\ \citenamefont
  {Young}(1993)}]{Young1993}%
  \BibitemOpen
  \bibfield  {author} {\bibinfo {author} {\bibfnamefont {M.}~\bibnamefont
  {Benedicks}}\ and\ \bibinfo {author} {\bibfnamefont {L.-S.}\ \bibnamefont
  {Young}},\ }\href@noop {} {\bibfield  {journal} {\bibinfo  {journal}
  {Inventiones mathematicae}\ }\textbf {\bibinfo {volume} {112}},\ \bibinfo
  {pages} {541} (\bibinfo {year} {1993})}\BibitemShut {NoStop}%
\bibitem [{\citenamefont {Osborne}\ and\ \citenamefont
  {Provenzale}(1989)}]{Osborne1989}%
  \BibitemOpen
  \bibfield  {author} {\bibinfo {author} {\bibfnamefont {A.~R.}\ \bibnamefont
  {Osborne}}\ and\ \bibinfo {author} {\bibfnamefont {A.}~\bibnamefont
  {Provenzale}},\ }\href@noop {} {\bibfield  {journal} {\bibinfo  {journal}
  {Physica D: Nonlinear Phenomena}\ }\textbf {\bibinfo {volume} {35}},\
  \bibinfo {pages} {357} (\bibinfo {year} {1989})}\BibitemShut {NoStop}%
\bibitem [{\citenamefont {Rapp}\ \emph {et~al.}(1993)\citenamefont {Rapp},
  \citenamefont {Albano}, \citenamefont {Schmah},\ and\ \citenamefont
  {Farwell}}]{Rapp1993}%
  \BibitemOpen
  \bibfield  {author} {\bibinfo {author} {\bibfnamefont {P.~E.}\ \bibnamefont
  {Rapp}}, \bibinfo {author} {\bibfnamefont {A.~M.}\ \bibnamefont {Albano}},
  \bibinfo {author} {\bibfnamefont {T.}~\bibnamefont {Schmah}}, \ and\ \bibinfo
  {author} {\bibfnamefont {L.}~\bibnamefont {Farwell}},\ }\href@noop {}
  {\bibfield  {journal} {\bibinfo  {journal} {Physical review E}\ }\textbf
  {\bibinfo {volume} {47}},\ \bibinfo {pages} {2289} (\bibinfo {year}
  {1993})}\BibitemShut {NoStop}%
\bibitem [{\citenamefont {Theiler}(1991)}]{Theiler1991}%
  \BibitemOpen
  \bibfield  {author} {\bibinfo {author} {\bibfnamefont {J.}~\bibnamefont
  {Theiler}},\ }\href@noop {} {\bibfield  {journal} {\bibinfo  {journal}
  {Physics letters A}\ }\textbf {\bibinfo {volume} {155}},\ \bibinfo {pages}
  {480} (\bibinfo {year} {1991})}\BibitemShut {NoStop}%
\bibitem [{\citenamefont {Schreiber}(1993)}]{Schreiber1993}%
  \BibitemOpen
  \bibfield  {author} {\bibinfo {author} {\bibfnamefont {T.}~\bibnamefont
  {Schreiber}},\ }\href {\doibase 10.1103/PhysRevE.48.R13} {\bibfield
  {journal} {\bibinfo  {journal} {Phys. Rev. E}\ }\textbf {\bibinfo {volume}
  {48}},\ \bibinfo {pages} {R13} (\bibinfo {year} {1993})}\BibitemShut
  {NoStop}%
\bibitem [{\citenamefont {Sauer}(1992)}]{Sauer1992}%
  \BibitemOpen
  \bibfield  {author} {\bibinfo {author} {\bibfnamefont {T.}~\bibnamefont
  {Sauer}},\ }\href@noop {} {\bibfield  {journal} {\bibinfo  {journal} {Physica
  D: Nonlinear Phenomena}\ }\textbf {\bibinfo {volume} {58}},\ \bibinfo {pages}
  {193} (\bibinfo {year} {1992})}\BibitemShut {NoStop}%
\bibitem [{\citenamefont {Kantz}\ \emph {et~al.}(1993)\citenamefont {Kantz},
  \citenamefont {Schreiber}, \citenamefont {Hoffmann}, \citenamefont {Buzug},
  \citenamefont {Pfister}, \citenamefont {Flepp}, \citenamefont {Simonet},
  \citenamefont {Badii},\ and\ \citenamefont {Brun}}]{Kantz1993}%
  \BibitemOpen
  \bibfield  {author} {\bibinfo {author} {\bibfnamefont {H.}~\bibnamefont
  {Kantz}}, \bibinfo {author} {\bibfnamefont {T.}~\bibnamefont {Schreiber}},
  \bibinfo {author} {\bibfnamefont {I.}~\bibnamefont {Hoffmann}}, \bibinfo
  {author} {\bibfnamefont {T.}~\bibnamefont {Buzug}}, \bibinfo {author}
  {\bibfnamefont {G.}~\bibnamefont {Pfister}}, \bibinfo {author} {\bibfnamefont
  {L.~G.}\ \bibnamefont {Flepp}}, \bibinfo {author} {\bibfnamefont
  {J.}~\bibnamefont {Simonet}}, \bibinfo {author} {\bibfnamefont
  {R.}~\bibnamefont {Badii}}, \ and\ \bibinfo {author} {\bibfnamefont
  {E.}~\bibnamefont {Brun}},\ }\href {\doibase 10.1103/PhysRevE.48.1529}
  {\bibfield  {journal} {\bibinfo  {journal} {Phys. Rev. E}\ }\textbf {\bibinfo
  {volume} {48}},\ \bibinfo {pages} {1529} (\bibinfo {year}
  {1993})}\BibitemShut {NoStop}%
\bibitem [{\citenamefont {Grassberger}\ \emph {et~al.}(1993)\citenamefont
  {Grassberger}, \citenamefont {Hegger}, \citenamefont {Kantz}, \citenamefont
  {Schaffrath},\ and\ \citenamefont {Schreiber}}]{Grassberger1993}%
  \BibitemOpen
  \bibfield  {author} {\bibinfo {author} {\bibfnamefont {P.}~\bibnamefont
  {Grassberger}}, \bibinfo {author} {\bibfnamefont {R.}~\bibnamefont {Hegger}},
  \bibinfo {author} {\bibfnamefont {H.}~\bibnamefont {Kantz}}, \bibinfo
  {author} {\bibfnamefont {C.}~\bibnamefont {Schaffrath}}, \ and\ \bibinfo
  {author} {\bibfnamefont {T.}~\bibnamefont {Schreiber}},\ }\href@noop {}
  {\bibfield  {journal} {\bibinfo  {journal} {Chaos: An Interdisciplinary
  Journal of Nonlinear Science}\ }\textbf {\bibinfo {volume} {3}} (\bibinfo
  {year} {1993})}\BibitemShut {NoStop}%
\end{thebibliography}%

\end{document}